\documentclass[a4paper, 10pt]{article}
\usepackage{amsmath}
\usepackage{amssymb}
\usepackage{graphicx}
\usepackage{geometry}
\usepackage[utf8]{inputenc}
\usepackage{hyperref}
\usepackage{xcolor}
\usepackage{comment}
\usepackage{rotating}
\usepackage[noend]{algpseudocode}
\usepackage[inline]{enumitem} 
\usepackage{graphicx} 
\usepackage{geometry}
\usepackage{amsfonts} 
\usepackage{float}
\usepackage{dsfont}
\usepackage{amsthm}
\usepackage{breakcites} 
\usepackage{tikz}
\usetikzlibrary{bayesnet}
\usepackage{amssymb}
\usepackage{authblk}
\usepackage{caption}
\usepackage{subcaption}
\usepackage{bm}
\usepackage[linesnumbered,lined,boxed,commentsnumbered,ruled,longend]{algorithm2e}
\usepackage{todonotes}
\usepackage{setspace}
\usepackage{multirow}
\usepackage{booktabs}
\usepackage{multicol} 

\PassOptionsToPackage{%
backend=bibtex,bibencoding=ascii,%
language=auto,%
style=authoryear-comp, 
bibstyle=authoryear,dashed=false, 
sorting=nyt,
maxbibnames=10, 
natbib=true 
,uniquename=init
,giveninits=true 
,url=true
}{biblatex}
\usepackage{biblatex}
\usepackage{tikz}
\usetikzlibrary{bayesnet}

\newcommand{\balpha}{\boldsymbol{\alpha}}
\newtheorem{theorem}{Theorem}
\newtheorem{lemma}[theorem]{Lemma}

\title{\textbf{A discomfort-informed adaptive Gibbs sampler for finite mixture models}}

\author[1,2]{Davide Fabbrico}
\author[3]{Andi Q. Wang}
\author[3,4]{Sebastiano Grazzi}
\author[3]{Alice Corbella}
\author[3]{Gareth~O.~Roberts}
\author[2,5]{Sylvia Richardson}
\author[2,3,*]{Filippo Pagani}
\author[2,*]{Paul~D.~W.~Kirk}
\date{\today}

\affil[1]{University of Florence, Department of Statistics, IT}
\affil[2]{MRC Biostatistics Unit, University of Cambridge, UK}
\affil[3]{University of Warwick, Department of Statistics, UK}
\affil[4]{Bocconi University, Department of Decision Sciences, IT }
\affil[5]{Integreat---Norwegian Centre for knowledge-driven machine learning, Department of Mathematics, University of Oslo, NO \vspace{0.2in}}

\affil[*]{Corresponding authors: Filippo Pagani and Paul D. W. Kirk, filippo.pagani@warwick.ac.uk, paul.kirk@mrc-bsu.cam.ac.uk}

\providecommand{\keywords}[1]
{
  \small	
  \textbf{\textit{Keywords---}} #1
}

\begin{document}

\maketitle

\begin{abstract}

\noindent Finite mixture models are frequently used to uncover latent structures in high-dimensional datasets (e.g.\ identifying clusters of patients in electronic health records). The inference of such structures can be performed in a Bayesian framework, and involves the use of sampling algorithms such as Gibbs samplers aimed at deriving posterior distribution of the probabilities of observations to belong to specific clusters. Unfortunately, traditional implementations of Gibbs samplers in this context often face critical challenges, such as inefficient use of computational resources and unnecessary updates for observations that are highly likely to remain in their current cluster. This paper introduces a new adaptive Gibbs sampler that improves the convergence efficiency over existing methods. 
In particular, our sampler is guided by a function that, at each iteration, uses the past of the chain to focus the updating on observations potentially misclassified in the current clustering, i.e.\ those with a low probability of belonging to their current component.
Through simulation studies and two real data analyses, we empirically demonstrate that, in terms of convergence time, our method tends to perform more efficiently compared to state-of-the-art approaches.
\end{abstract}


\section{Introduction}

Mixture models provide a flexible and widely-used framework for modelling data, where each observation is assumed to originate from one of up to $K$ mixture components (also known as \emph{clusters}), with $K$ potentially unknown. Each cluster is represented by a distribution function that belongs to a specified parametric family. Commonly used model classes include finite mixtures \citep[e.g.][]{Pearson1894,Everitt1981,Titterington1985,McLachlan1987}, mixtures of finite mixtures \parencite{Richardson1997,nobile2004,miller2018}, and Dirichlet process (\textcite{maceachern1994}, \textcite{escobar1995}, \textcite{neal2000}). In the former case, the number of components is fixed and specified a priori. In the latter two cases, the number of components can vary and adapt to the complexity of the dataset.

These models are commonly used in many fields, such as genomics \citep[e.g.][]{broet2002bayesian, fraley2002model} and epidemiology \citep[e.g.][]{schlattmann1993mixture, green2002hidden}. In econometrics, mixture models have proven valuable in areas such as marketing \citep[e.g.][]{jedidi1997finite, allenby1999dynamic}, macroeconomics \citep[e.g.][]{hamilton1989new, fruhwirth2001markov}, and finance \citep[e.g.][]{fruhwirth2001markov, weigend2000predicting, kaufmann2002bayesian}.

Although the modelling flexibility of mixture models is well-established, efficient computational inference remains a significant challenge, especially for large datasets where traditional approaches may become prohibitively expensive. Moreover, mixture models present particular sampling difficulties, including multi-modality and label switching (i.e.\ the likelihood being invariant to permutations of the cluster labels) \citep{fruhwirth2006finite, stephens2000label}. Parameter inference is often performed using Markov chain Monte Carlo (MCMC) algorithms, particularly the Gibbs sampler \citep{gelfand1990sampling}, which alternately updates model parameters $(\boldsymbol{\Theta}, \boldsymbol{\pi})$ and latent allocation variables $\boldsymbol{z} = (z_1, \ldots, z_n)$, where $z_i \in \{1,2,\ldots,K\}$ indicates the cluster assignment of the $i$-th observation and $n$ is the sample size.

The two most commonly used methods are \emph{Systematic Scan Gibbs} (SSG) and \emph{Random Scan Gibbs} (RSG) \citep{geman1984stochastic}. SSG updates all $n$ allocations in a predetermined fixed order, while RSG randomly selects a subset of allocations to update at each iteration \citep{amit1991comparing, fishman1996coordinate, roberts1997updating}. Both algorithms have been extensively analysed, with mixed performance depending on the target density \citep{roberts2015surprising, he2016scan, Ascolani2024, Goyal2025}.

However, both approaches spend computational effort uniformly across observations, regardless of how confidently each observation is assigned to its current cluster. This uniform updating scheme is computationally wasteful: in typical clustering scenarios, many observations may be confidently assigned, while only a subset near cluster boundaries would benefit from frequent reallocation attempts. 

To address such inefficiencies, {\em adaptive} MCMC approaches learn and optimise sampling strategies during algorithm execution \citep{haario2001adaptive, roberts2007coupling, andrieu2008tutorial}. These adaptive methods leverage information about the target distribution that accumulates over time to dynamically adjust the sampling mechanism, and include the use of non-uniform selection probabilities that have been shown to improve convergence rates \citep{latuszynski2013adaptive}. 
For mixture models specifically, adaptive approaches have been developed to address label switching \citep{bardenet2015adaptive} and to implement efficient block sampling strategies \citep{fiorentini2014efficient}.

This manuscript introduces a novel ``discomfort-informed'' adaptive Gibbs sampler that addresses the inefficiency of uniform updating in mixture models. Here, `discomfort' quantifies how poorly suited an observation is to its current cluster assignment---observations with low posterior probability of belonging to their assigned cluster receive high discomfort scores and are prioritised for updates, while confidently-assigned observations are updated less frequently. Our approach dynamically identifies observations with high uncertainty about their current cluster assignment and prioritises updating these observations over those that are confidently assigned. 

Our adaptive scheme maintains rigorous theoretical foundations, satisfying established conditions for ergodicity of adaptive MCMC algorithms \citep{roberts2007coupling}. Our selective updating strategy leads to substantial improvements in convergence speed---often achieving 5--10$\times$ faster convergence than standard methods---while maintaining the theoretical guarantees of traditional Gibbs sampling.

The paper is structured as follows. Section \ref{sec:fmm} introduces the framework for finite mixture models and reviews key algorithms commonly used for parameter inference, including a generic adaptive algorithm, along with their limitations. It also reviews some canonical sufficient conditions establishing the ergodicity of underlying stochastic processes. Section \ref{sec:discomfort} discusses the rationale and choices underlying our adaptive MCMC algorithm and provides theoretical findings on the convergence of the selection probability vector. Section \ref{sec:numExp} evaluates finite-sample performance through simulation studies, benchmarking our approach against state-of-the-art models. In Section \ref{sec:real}, we apply our methodology to diverse datasets, demonstrating its inferential capabilities and interpretability. Finally, Section \ref{sec:disc} concludes the paper with a brief discussion. Further technical discussions and proofs are included in the appendices. Working code can be found at \href{https://github.com/davidefabbrico/AdaptiveAllocation}{github.com/davidefabbrico/AdaptiveAllocation}.

\section{Finite mixture models} \label{sec:fmm}

Let the data be represented by the matrix $\bm X \in \mathbb{R}^{n \times d}$, where $d$ represents the number of variables, $n$ represents the sample size, and the $i$-th row is denoted by the vector $\bm x_i=(x_{i1}, x_{i2}, \ldots, x_{id})$, for $i = 1, \ldots, n$. We model it using a finite mixture with $K$ components and parameters $\bm \Theta = (\bm \theta_1, \bm \theta_2, \ldots, \bm \theta_K)$, which can be written as: 
\begin{equation}
    p(\bm X \mid \bm \Theta, \bm \pi) = \prod_{i=1}^n p(\boldsymbol{x}_i \mid \bm \Theta, \bm \pi) \, ,
\end{equation}
where, 
\begin{equation*}
    p(\bm x_i \mid \bm \Theta, \bm \pi) = \sum_{k=1}^K \pi_k f_x(\bm x_i \mid \bm \theta_k) \, ,
\end{equation*}
with
\begin{align*}
\bm{x}_i \mid z_i, \boldsymbol{\theta} &\sim F_x(\boldsymbol{\theta}_{z_i}),\notag\\
z_i \mid \boldsymbol{\pi}&\sim \mbox{Categorical}(\pi_1, \ldots, \pi_K) \notag \, ,\\
\pi_1, \ldots, \pi_K & \sim \mbox{Dirichlet}(a/K, \ldots, a/K) \label{2} \, , \\
\boldsymbol{\theta}_1 , \ldots, \boldsymbol{\theta}_K & \sim G^{(0)} \, . \notag
\end{align*}
Here $F_x$ is a generic distribution with density $f_x(\boldsymbol{x}_i \mid \bm \theta_k)$, defining the component-specific model. The latent variable $\boldsymbol{z} = (z_1, z_2, ..., z_n)$, where each $z_i$ takes values in $\{1,2,\ldots,K\}$, represents the component allocation of individual $i$, while $\boldsymbol{\pi} = \left(\pi_1, \ldots, \pi_K\right)$ is a collection of $K$ component weights. Finally, the scalar $a$ is a mass/concentration parameter, here assumed fixed, and $G^{(0)}$ is the prior for the component-specific parameters.

\subsection{Standard inference approaches}
To infer the unknown quantities $\boldsymbol{z}$, $\boldsymbol{\pi}$, and $\boldsymbol{\Theta}$ given the data $\boldsymbol{X}$, we employ a blocked Gibbs sampling strategy where these parameter blocks are updated sequentially at each iteration. The two most commonly used approaches are the Systematic Scan Gibbs (SSG) and Random Scan Gibbs (RSG) samplers, as described in Algorithms \ref{alg:ssgibbs} and \ref{alg:rsgibbs}, respectively.

The SSG method updates allocation variables for all $n$ observations in a predetermined order at each iteration, followed by updates to $\boldsymbol{\pi}$ and $\boldsymbol{\Theta}$. In contrast, the RSG approach updates allocation variables for only a randomly selected subset of $m \ll n$ observations at each iteration, where the subset is chosen uniformly at random with replacement from $\{1, \ldots, n\}$.

\begin{multicols}{2}
\begin{algorithm}[H]
\SetKwInOut{Input}{input}\SetKwInOut{Output}{output}
\SetAlgoLined
\Input{$T$ MCMC iterations}
\BlankLine
\For{$t = 1, \ldots, T$}{
\protect\phantom{$i_1, \ldots, i_m \sim \text{Uniform}(\{1, \ldots, n\})$}\\
\For{$i = 1, \ldots, n$}{
$z_{i,t} \sim p(z_i \mid \boldsymbol{\Theta}_{t-1} , \boldsymbol{\pi}_{t-1}, \boldsymbol{X})$ 
}
$\boldsymbol{\pi}_t \sim p(\boldsymbol{\pi} \mid \boldsymbol{\Theta}_{t-1} ,\boldsymbol{z}_{t} , \boldsymbol{X})$ \\
$\boldsymbol{\Theta}_{t} \sim p(\boldsymbol{\Theta} \mid \boldsymbol{\pi}_{t} ,\boldsymbol{z}_{t} , \boldsymbol{X})$ 
}
\caption{\small SSG}
\label{alg:ssgibbs}
\end{algorithm} 
\columnbreak 
\begin{algorithm}[H]
\SetKwInOut{Input}{input}\SetKwInOut{Output}{output}
\SetAlgoLined
\Input{$T$ iterations, subset size $m$}
\BlankLine
\For{$t = 1, \ldots, T$}{
$i_1, \ldots, i_m \overset{\text{iid}}{\sim} \text{Uniform}(\{1, \ldots, n\})$\\
\For{$j = 1, \ldots, m$}{
$ z_{i_j,t} \sim p(z_{i_j} \mid \boldsymbol{\Theta}_{t-1} , \boldsymbol{\pi}_{t-1}, \boldsymbol{X} )$ 
}
$\boldsymbol{\pi}_{t} \sim p(\boldsymbol{\pi} \mid \boldsymbol{\Theta}_{t-1} ,\boldsymbol{z}_{t} , \boldsymbol{X})$ \\
$\boldsymbol{\Theta}_{t} \sim p(\boldsymbol{\Theta} \mid \boldsymbol{\pi}_{t} ,\boldsymbol{z}_{t} , \boldsymbol{X})$ 
}
\caption{\small RSG} 
\label{alg:rsgibbs}
\end{algorithm}
\end{multicols}

\subsection{Computational limitations of standard approaches}

Both SSG and RSG suffer from inefficiencies that become more pronounced as dataset size increases.

\textbf{SSG}: When $n$ is large, computing the full conditional distributions for all allocation variables becomes computationally expensive, requiring $\mathcal{O}(nK)$ likelihood evaluations per iteration. Moreover, the algorithm wastes computational effort on observations that are already confidently assigned to their correct component and unlikely to switch, leading to redundant updates that contribute little to posterior exploration.

\textbf{RSG}: Although computationally cheaper per iteration with only $\mathcal{O}(mK)$ likelihood evaluations, RSG will require approximately $n/m$ iterations to update all observations once. Uniform random selection means that computational resources are still wasted on confidently assigned observations that have a low probability of changing their current allocation.

Both methods ignore the fact that in typical clustering analyses, observations will have varying levels of assignment uncertainty. Observations near cluster centres will be confidently assigned and rarely benefit from reallocation attempts, while the allocation of observations near cluster boundaries will be uncertain and would benefit from frequent updates. This uniform treatment represents a fundamental mismatch between computational effort and inferential need.

\subsection{A motivating example where standard methods do not converge} \label{sec:motivating}

To illustrate these limitations, we consider a finite mixture model with $K = 5$ components, and dimension $d = 2$. 
The component weight vector $\bm{\pi}$ is uniform, i.e.\ $\pi_k = 1/5$ for $k = 1, \ldots, 5$, with prior $\bm{\pi} \sim \text{Dirichlet}(a_1/K, \ldots, a_5/K)$, with $a_k = 1$ for all $k$.
Each component follows a bivariate normal distribution $f_x(\bm{x}_i \mid \bm{\theta}_k) = \mathcal{N}(\bm{\mu}_k, \bm{\Sigma}_k)$, where $\bm \Sigma_k = \mathds{I}_{d \times d}$ for all $k$, and $\bm{\theta}_k = (\bm{\mu}_k, \bm{\Sigma}_k)$. The specific mean vectors are:
\begin{equation*}
    \bm{\mu}_1 = (-10,-10), \
    \bm{\mu}_2 = (-5,-5),\ 
    \bm{\mu}_3 = (0,0),\ 
    \bm{\mu}_4 = (5,5),\ 
    \bm{\mu}_5 = (10,10) \, ,
\end{equation*}
so that the clusters all lie on the diagonal between the first and third quadrant of the $x-y$ plane, similarly to \citet{miller2018}, Section 7.1.2. The chosen parameter settings ensure that the clusters are well-separated in this example.
The component means are assigned independent Gaussian priors, $\bm{\mu}_k \sim \mathcal{N}(\bm{0}, \sigma^2 \, \mathds{I}_{2})$, where we generated the data with $\sigma^2 = 1$ for all components.
We assign to $\sigma^2$ an Inverse-Gamma prior, namely $\sigma^2 \sim \mathcal{IG}(\alpha_\sigma, \beta_\sigma)$, where conditional on the data, the hyperparameters $(\alpha_\sigma, \beta_\sigma)$ are set through the Empirical Bayes regularization method proposed by \citet{fraley2007bayesian}.

From this model we draw a dataset of $n = 500$ samples and study the performance of the algorithms described in the previous section. We run the comparison both on correctly specified models (where the number of components in the fitted model $K^{\ast}$ is equal to the number of components used to generate the data $K$, in this case $K^{\ast} = K = 5$) and overfitted models (in this case $K = 5$, and $K^{\ast} = 10$), the latter being especially relevant for practical applications where the number of clusters is unknown \citep{rousseau2011}. Well-performing algorithms should naturally identify the correct number of occupied clusters even when $K^{\ast} > K$, effectively performing model selection alongside parameter estimation. 

To ensure fair comparison, we initialise SSG, RSG, and our proposed {\em Discomfort-Informed Gibbs} (DIG) adaptive method---described in detail in Section \ref{sec:discomfort}---with identical starting values, and run all algorithms for 5000 iterations without burn-in or thinning across 20 independent replicas. 
Similarly to RSG, DIG updates only a subset of the observations of size $m \le n$ (here $m = 5$) at each iteration. Therefore DIG is computationally comparable to RSG, whereas SSG is generally more computationally intensive since $m \ll n$. 
In the numerical comparisons of Sections~\ref{sec:numExp}–\ref{sec:real}, this factor is taken into account by normalising the performance metrics considered either by computing time or by the number of epochs.
However, in this section we ignore the difference in computational cost, as our objective is to show that on a simple problem, standard algorithms may fail to recover the correct clustering structure regardless of the computational budget.

Let $\phi(\bm{x} \mid \bm{\mu}, \bm{\Sigma})$ be the multivariate Gaussian density with mean $\bm{\mu}$ and covariance $\bm{\Sigma}$ evaluated at $\bm{x}$. 
The Complete Log-Likelihood (CLL) at iteration $t$ is defined as
\begin{equation}
\label{eq:CLL}
\text{CLL}^t(\boldsymbol{z}, \bm{\mu}, \bm{\Sigma}, \bm{\pi}) = \sum_{i=1}^n \log \left( \hat{\pi}_{z_i}^t \phi(\bm{x}_i \mid \hat{\bm{\mu}}_{z_i}^t, \hat{\bm{\Sigma}}_{z_i}^t) \right) \, ,
\end{equation}
where $(\hat{\pi}_k, \hat{\bm{\mu}}_k, \hat{\bm{\Sigma}}_k)$ denote running estimates of $(\pi_k, \bm{\mu}_k, \bm{\Sigma}_k)$ obtained from the chain up to iteration $t$. Below we show convergence behaviour based on calculating the CLL at each iteration, and we use it as a diagnostic tool and as an internal measure of clustering quality.

\begin{figure}[htbp]
    \centering
    \begin{subfigure}[b]{0.48\textwidth}
        \includegraphics[width=\textwidth]{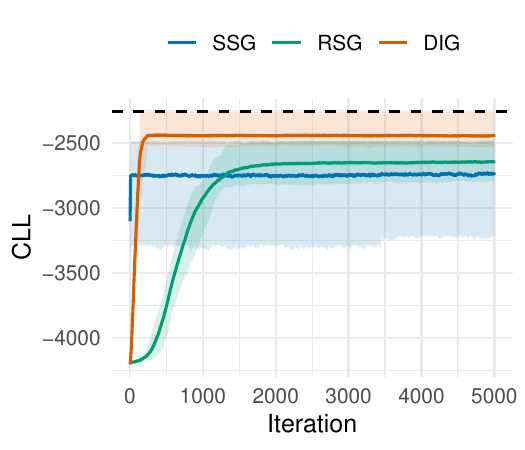}
        \caption{Convergence behaviour with correctly specified number of components: $K^{\ast} = K = 5$.}
        \label{fig:noConv_left}
    \end{subfigure}
    \hfill
    \begin{subfigure}[b]{0.48\textwidth}
        \includegraphics[width=\textwidth]{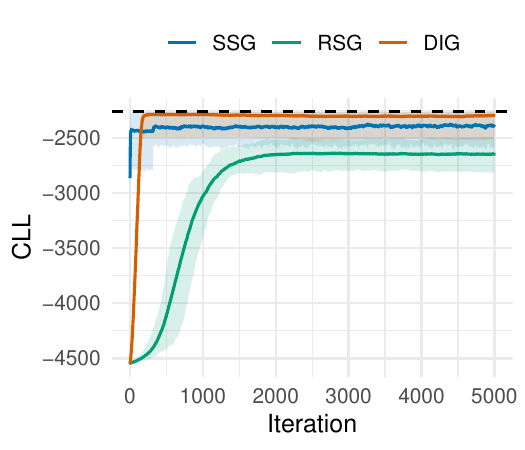}
        \caption{Convergence behaviour with over-specified number of components: $K^{\ast} = 10 > K = 5$.}
        \label{fig:noConv_right}
    \end{subfigure}
    \caption{Motivating Example: finite mixture with $K = 5$ components ($n = 500, \, d = 2$). Complete log-likelihood (CLL) for SSG, RSG and DIG methods. Solid lines indicate the mean across 20 independent replicas, with shaded regions representing the 95\% credible interval. A horizontal dashed black line indicates the CLL calculated using the true parameter values.}
    \label{fig:noConv}
\end{figure}

Figure \ref{fig:noConv} shows that for the example above, SSG stabilises quickly but becomes trapped at a substantially lower CLL in both correctly specified ($K^{\ast} = K = 5$) and overfitted ($K^{\ast} = 10 > K = 5$) scenarios, indicating poor mixing and inadequate posterior exploration. RSG achieves better CLL values than SSG in the correctly specified case, but generally converges more slowly, while DIG reaches the highest values most rapidly in both conditions. A dashed black horizontal line in the figure indicates the true CLL, computed from Equation \eqref{eq:CLL} using the true parameter values.
Importantly, when initialised from the true parameters and cluster allocations, both SSG and RSG maintain a CLL value comparable to DIG. This confirms that the cause of SSG and RSG stabilising at a lower CLL value than DIG in earlier experiments is inadequate mixing, rather than fundamental algorithmic variance.

This example demonstrates that SSG's deterministic, full-component updates can amplify autocorrelation and hinder exploration of the posterior distribution, while uniform updating strategies in general can lead to inadequate mixing. This motivates the need for adaptive approaches that intelligently prioritise observations with high assignment uncertainty.

\subsection{Adaptive Gibbs samplers for finite mixture models} \label{sec:adaptive}

A general adaptive Gibbs algorithm maintains the same structure as the RSG, but replaces uniform selection probabilities with adaptive weights $\boldsymbol{\alpha}_t = (\alpha_{1t}, \ldots, \alpha_{nt})$ that evolve according to some update rule $R_t$. Additionally, we sample without replacement to ensure each selected observation is distinct within each iteration, as shown in Algorithm \ref{alg:ARSG}.

\begin{algorithm}[!ht]
\SetKwInOut{Input}{input}\SetKwInOut{Output}{output}
\SetAlgoLined
\Input{Initial sampling weights $ \boldsymbol{\alpha}_0= (\alpha_{01}, \ldots, \alpha_{0n}) $; $ T $ iterations; $m$}
\BlankLine
\For{$t = 1, \ldots, T$}{
$\boldsymbol{\alpha}_t := R_t(\boldsymbol{\alpha}_0 , \ldots, \boldsymbol{\alpha}_{t-1}, \boldsymbol{z}_0, \ldots, \boldsymbol{z}_{t-1})$ \hfill{\scriptsize update the sampling weights}\\
Sample $m$ distinct indices $\{i_1, \ldots, i_m\}$ with probabilities $\boldsymbol{\alpha}_t$ \hfill{\scriptsize sample w/o replacement }\\
\For{$j = 1, \ldots, m$}{
$z_{i_j,t} \sim p(z_{i_j} \mid \boldsymbol{\Theta}_{t-1} , \boldsymbol{\pi}_{t-1}, \boldsymbol{X})$ \hfill{\scriptsize update the selected latent variable}
}
$\boldsymbol{\pi}_{t} \sim p(\boldsymbol{\pi}\mid\boldsymbol{\Theta}_{t-1} ,\boldsymbol{z}_{t} , \boldsymbol{X})$ \hfill{\scriptsize update the mixture weights}\\
$\boldsymbol{\Theta}_{t} \sim p(\boldsymbol{\Theta}\mid\boldsymbol{\pi}_{t} ,\boldsymbol{z}_{t} , \boldsymbol{X})$ \hfill{\scriptsize update the component profiles }
}
\caption{\small ARSG: Adaptive Random Scan Gibbs for finite mixture models}
\label{alg:ARSG}
\end{algorithm}

To develop a practical adaptive sampler based on Algorithm \ref{alg:ARSG}, we need to ensure ergodicity and develop a computationally tractable update rule $R_t$ that upweights observations that have high allocation uncertainty.

\subsection{Ergodicity of adaptive algorithms} \label{sec:theory}

For adaptive MCMC algorithms to be valid, they must satisfy conditions ensuring convergence to the correct target distribution. \textcite{latuszynski2013adaptive} provide explicit sufficient conditions for adaptive MCMC ergodicity that apply broadly to adaptive algorithms, including Algorithm~\ref{alg:ARSG} for finite mixtures.  In this context, the theorem is as follows.

\begin{theorem} \label{AdTheorem}
Let $\mathcal{Y}$ be the space of possible selection probability vectors $\boldsymbol{\alpha}_t$. Following \textcite{latuszynski2013adaptive}, assume:
\begin{itemize}
\item[(a)] \textbf{Diminishing adaptation}: $|\boldsymbol{\alpha}_t-\boldsymbol{\alpha}_{t-1}| \to 0$ in probability for fixed starting values $\boldsymbol{z}_0, \boldsymbol{\Theta}_0, \boldsymbol{\pi}_0$ and $\boldsymbol{\alpha}_0 \in \mathcal{Y}$;
\item[(b)] \textbf{Containment}: There exists $\boldsymbol{\beta} \in \mathcal{Y}$ such that ARSG with fixed probabilities $\boldsymbol{\alpha}_t = \boldsymbol{\beta}$ is uniformly ergodic.
\end{itemize}
Then the adaptive method is ergodic:
\begin{equation*}
\| P_t(\boldsymbol{z}, \boldsymbol{\theta} \mid \boldsymbol{\alpha}_0, \boldsymbol{z}_0, \boldsymbol{\Theta}_0, \boldsymbol{\pi}_0) - \nu(\boldsymbol{z}, \boldsymbol{\theta})\|_{\mathrm{TV}} \to 0, \qquad \text{as } t\to \infty \, ,
\end{equation*}
where $P_t$ denotes the distribution at time $t$, $\nu$ is the target posterior distribution, and $TV$ stands for Total Variation.
\end{theorem}

\textcite{diebolt1994estimation} establish that condition (b) is readily satisfied in mixture models. They show that the augmented chain $\{\boldsymbol{\theta}_t, \boldsymbol{z}_t\}$ is uniformly ergodic provided the allocation chain $\{\boldsymbol{z}_t\}$ is uniformly ergodic. Since $\boldsymbol{z}_t$ evolves on the finite space $\{1,2,\ldots,K\}^n$, uniform ergodicity follows from irreducibility and aperiodicity. For mixture models, this occurs when all components of the selection probabilities remain bounded away from zero. In particular, the uniform selection probabilities $\boldsymbol{\beta} = (1/n, \ldots, 1/n)$ satisfy this condition, ensuring containment is met.

With containment established, the key requirement for an adaptive algorithm is that adaptations decay over time to satisfy the diminishing adaptation condition. In the next section, we develop the Discomfort-informed Gibbs (DIG) sampler, which satisfies this condition whilst efficiently targeting observations with high assignment uncertainty.

\section{Discomfort-informed Gibbs Sampler}
\label{sec:discomfort}

In this section, we provide an overview of the choices underlying our adaptive algorithm. Section~\ref{sec:alpha_update} introduces the adaptive update rule $R_t$ for the selection weights $\boldsymbol{\alpha}_t$, while Section~\ref{discomfort} defines the discomfort function, which distinguishes between ``interesting'' and less informative observations for updating purposes. Section~\ref{weights} presents the weight functions that regulate the trade-off between exploration and exploitation. Section~\ref{sec:stoc} establishes the convergence of the selection probability vector through the use of stochastic approximation method. Finally, Section~\ref{sec:completeAlgo} details the complete algorithm.

\subsection{Update function \texorpdfstring{$R_t$}{Rt}}
\label{sec:alpha_update}

Let $\bm D_t^{\lambda_t} = (D_{1t}^{\lambda_t}, \ldots, D_{nt}^{\lambda_t})$ denote the vector of {\em discomfort} values for each observation at iteration $t$ (defined below in Section~\ref{discomfort}, Equation~\eqref{DivFun}), and $f(t)$, $g(t)$ be non-negative weight functions satisfying $f(t) + g(t) = 1$ for all $t$. 
We consider adaptive update functions, $R_t$ of the following form:
\begin{equation}
\boldsymbol{\alpha}_t = R_t(\boldsymbol{\alpha}_0 , \ldots, \boldsymbol{\alpha}_{t-1}, \boldsymbol{z}_0, \ldots, \boldsymbol{z}_{t-1}) = f(t) \cdot \bm{\alpha}_{t-1} + g(t) \cdot \bm D_t^{\lambda_t} \, ,
\label{eq:adaptrule}
\end{equation} 
where after computing $\boldsymbol{\alpha}_t$, we normalise to obtain a valid probability vector: $\alpha_{i,t}' = \alpha_{i,t} / \sum_{j=1}^n \alpha_{j,t}$.
To satisfy the diminishing adaptation condition (a) from Theorem~\ref{AdTheorem}, we require $f(t) \to 1$ and $g(t) \to 0$ as $t \to \infty$, ensuring that the influence of new information vanishes over time. We discuss these weight functions in Section~\ref{weights}.

\subsection{Discomfort function \texorpdfstring{$D^\lambda_{it}$}{Dt}}\label{discomfort}

The discomfort function serves as our mechanism for distinguishing between observations that warrant attention and those that do not need to be prioritised for allocation updating, and serves as a key component of our update function, $R_t$. Intuitively, an observation with high posterior probability of belonging to its current cluster should have low discomfort, while an observation poorly matched to its current assignment should have high discomfort.

We define the discomfort function using the negative exponential form:
\begin{equation}
D^\lambda_{it}=D^\lambda_{it}(\boldsymbol{\theta}, \boldsymbol\pi, \boldsymbol{z}) = \exp\left({-\lambda \,p_{i,z_i^t}^t}\right) , \quad \text{with }
     p_{i,z_i^t}^{t}:= \frac{\pi_{z_{i,t}} f_x\left (\bm{x}_i\mid \bm{\theta}_{z_{i,t}} \right) }{\sum_{k=1}^K\pi_{k} f_x \left (\bm{x}_i\mid \bm{\theta}_{k}\right)} \, ,
    \label{DivFun}
\end{equation}
where $\lambda>0$ represents the discomfort decay parameter, and $p_{i,z_i^t}^t$ is the posterior probability that observation $i$ belongs to its currently assigned component $z_i^t$ at iteration $t$.

This exponential form ensures that $D^\lambda_{it}$ is a decreasing function of $p_{i, z_i^t}^t$: observations with lower assignment probabilities receive higher discomfort scores and thus higher selection priority. The parameter $\lambda$ controls the sensitivity of this relationship---larger values of $\lambda$ create sharper distinctions between well-assigned and poorly-assigned observations.

Although alternative functions such as entropy-based measures or heavy-tailed distributions could also be considered (see Appendix \ref{alternatives}), the exponential form offers computational simplicity and a natural connection to statistical mechanics frameworks.  In particular, this choice has the same functional form as unnormalised Boltzmann-Gibbs weights from statistical mechanics \citep{pathria2011statistical}, where $p_i$ acts analogously to energy and $\lambda$ corresponds to inverse temperature. This analogy provides some intuition: at high ``temperatures'' (low $\lambda$), the selection mechanism becomes relatively uniform across observations, while at low temperatures (high $\lambda$), selection concentrates sharply on low-energy (poorly-assigned) observations.

Computing the discomfort vector $\bm D_t^{\lambda_t}$ requires the allocation probability matrix; an $n \times K$ matrix where entry $(i,k)$ gives the posterior probability that observation $i$ belongs to component $k$. Updating this matrix requires a full dataset pass, however, since $m \ll n$, the allocation matrix does not change appreciably from one iteration to the next.  We therefore update it only periodically, as described in Section \ref{allUp}.

\subsubsection{Adaptive discomfort decay parameter \texorpdfstring{$\lambda_t$}{lt}} \label{sec:disDec}
The parameter $\lambda$ in Equation~\eqref{DivFun} determines the discomfort function's behaviour, as illustrated in Figure~\ref{fig:lambda}. 
\begin{figure}[!ht]
    \centering
    \includegraphics[width=\linewidth]{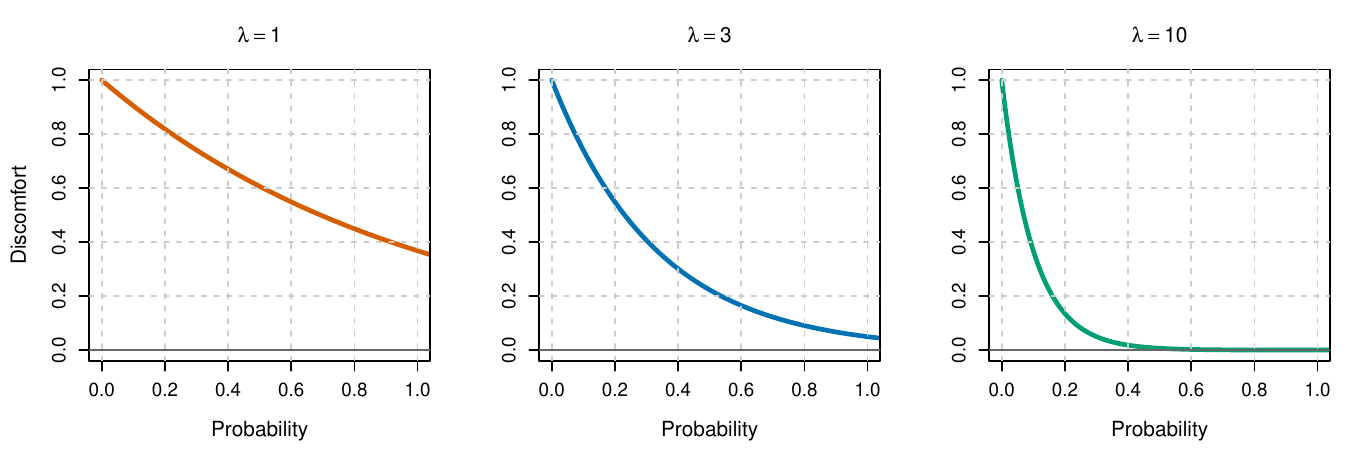}
    \caption{Exponential discomfort functions $D^{\lambda}$ for three values of the tuning parameter $\lambda \in \{1, 3, 10\}$. The discomfort $D^{\lambda}$ (on the $y$-axis) is plotted as a function of the assignment probability $p \in [0,1]$ (on the $x$-axis), showing how higher values of $\lambda$ produce sharper distinctions between well-assigned ($p \approx 1$) and poorly-assigned ($p \approx 0$) observations.}
    \label{fig:lambda}
\end{figure}
The choice of $\lambda$ creates a trade-off that can be seen by examining different regions of the probability range. High values of $\lambda$ (e.g.\ $\lambda = 10$) yield a discomfort function that strongly emphasises observations with low assignment probabilities, which is desirable for focusing on poorly-allocated cases. However, the portion of the curve between approximately $0.6$ and $1.0$ becomes nearly flat, which becomes a problem after the most poorly allocated observations have been reallocated, since the algorithm loses sensitivity to differences in assignment probabilities among reasonably well-assigned observations.

Conversely, small values of $\lambda$ (e.g.\ $\lambda = 1$) maintain sufficient slope in the higher probability region ($0.6$ to $1.0$) to distinguish among well-assigned observations. However, the lower probability region ($0$ to $0.6$) lacks the steepness needed for the algorithm to focus sharply on the most poorly-assigned observations.

This trade-off is further complicated by the data structure: well-separated clusters benefit from higher $\lambda$ values that emphasise the clear distinctions, while overlapping clusters require more moderate values to maintain sensitivity to ambiguously assigned observations. Moreover, the optimal $\lambda$ may change throughout the MCMC run as the allocation quality improves, motivating our adaptive approach described below.

Rather than attempt to specify $\lambda$ a priori, we adapt it dynamically using an effective sample size (ESS) criterion. We borrow the ESS formula from sequential Monte Carlo and apply it to our discomfort values to measure how concentrated the selection mechanism has become:
\begin{equation}
\label{eq:ESS}
\mathrm{ESS}(\lambda) = \frac{ \left( \sum_{i=1}^n \exp \left \{-\lambda p_{i,z_i^t}^t \right \} \right)^2 }{\sum_{i=1}^n \exp \left \{-2\lambda p_{i,z_i^t}^t \right \} } \, .
\end{equation}

Since we update $m$ observations per iteration, we seek the value $\hat{\lambda}$ that yields $\mathrm{ESS}(\lambda) \approx m$. This choice ensures that the selection probabilities are concentrated enough to focus on problematic observations, while maintaining sufficient diversity to avoid over-concentration on just a few instances.

This approach is heuristic but we have found it to work well in practice, automatically balancing exploration and exploitation as the algorithm progresses. Early iterations typically yield high ESS values requiring larger $\lambda$, while later iterations produce lower ESS values with correspondingly smaller $\lambda$.

For theoretical and practical validity, we constrain $\lambda_t$ to the compact interval $[1, \Lambda]$. This ensures that all selection probabilities remain bounded away from zero, preventing the algorithm from permanently excluding any observations from consideration, while also providing numerical stability for the root-finding procedure. In practice, we find $\Lambda = 100$ to work well for most applications. More details on how to select $\Lambda$ are given in Appendix \ref{app:tuning}.

The resulting adaptation schedule is:
\begin{equation}
\label{eq:lambda_schedule}
\lambda_t = 
\begin{cases} 
\hat{\lambda}_t, & t \leq s \\
1, & t > s \, 
\end{cases}, \qquad{\hat \lambda_t = \{\lambda \in [1, \Lambda] \colon \mathrm{ESS}(\lambda) - m  = 0 \}},
\end{equation}
where $s$ is a pre-specified iteration discussed in Section \ref{sec:s} and $\hat \lambda_t$ is found numerically with a root-finder method.

\subsubsection{The allocation matrix update schedule \texorpdfstring{$\xi$}{xi}} \label{allUp}
The allocation probability matrix changes slowly when $m$ is small relative to $n$, making frequent updates unnecessary and computationally wasteful. We therefore update the matrix every $\xi$ iterations, where $\xi$ varies through the course of running the algorithm. 
During the initial ``greedy'' phase when allocations change rapidly, we update frequently ($\xi = 3$ for the first 25\% of iterations). As the algorithm stabilises, we gradually increase the interval: $\xi = 6$ for iterations 25-50\%, and $\xi = 10$ for the final 50\%. 
This schedule is similar to those used in popular probabilistic computing software \citep{carpenter2017stan}, and more details and a sensitivity analysis can be found in Appendix \ref{app:free}.

This adaptive schedule maintains the correct stationary distribution since the adaptation is diminishing in the sense required by our theoretical framework, and aligns with common practice in adaptive MCMC literature \citep[e.g.][]{haario2001adaptive, roberts2009examples}.

\subsection{Weight functions \texorpdfstring{$f(t)$}{ft} and \texorpdfstring{$g(t)$}{gt}}\label{weights}

The choice of weight functions determines the algorithm's exploration-exploitation balance. We employ a two-stage approach designed to promote initial exploration during stage 1, until a switch at iteration $s$ when we enter stage 2 focused on promoting stable convergence. The value of $s$ must be pre-specified; see Section~\ref{sec:s} for guidance.

For iterations $t \leq s$, we use hyperbolic tangent functions:
\begin{equation}
\label{eq:wFunHyper}
    f(t) = \frac{1}{2} \left(1+\tanh{\left(\frac{t-s}{a} \right)}\right) \hspace{0.1cm}, \hspace{0.3cm} g(t) = \frac{1}{2} \left(1-\tanh{\left(\frac{t-s}{a} \right)}\right) \, ,
\end{equation}
where $a$ controls how rapidly the functions transition between stages and $s$ marks the switch point. This first stage emphasises the discomfort-driven term $g(t)$, strongly promoting reallocation of poorly-assigned observations.

For iterations $t > s$, we switch to polynomial decay:
\begin{equation} 
\label{eq:wFunPol}
    f(t) = \frac{t-s+1}{t-s+2} \hspace{0.1cm}, \hspace{0.3cm} g(t) = \frac{1}{t-s+2} \, ,
\end{equation}
which satisfies the diminishing adaptation requirement.

Figure \ref{fig:weightfunc} illustrates the behaviour of our weight functions before and after the point $s$.

\begin{figure}[!ht]
    \centering
    \includegraphics[scale = 0.62]{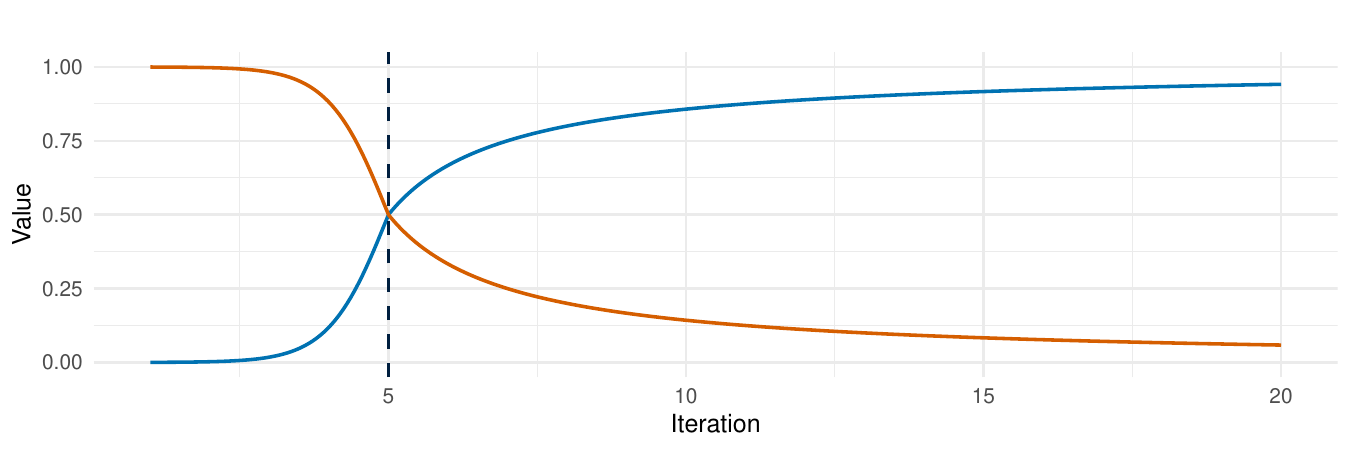}
    \caption{Visualisation of weight functions $f(t)$ (in blue) and $g(t)$ (in orange) before and after the point $s$ (dotted vertical line). For $t \leq s$, the functions are defined by the hyperbolic tangent equations (Equation \eqref{eq:wFunHyper}), while for $t > s$, the weights are governed by the polynomial equations (Equation \eqref{eq:wFunPol}). In this specific plot, $a = 1$ and $s = 5$. The transition at $s = 5$ marks the switch from a hyperbolic decay (exploration phase) to a polynomial decay (convergence phase).}
    \label{fig:weightfunc}
\end{figure}

We discuss alternative options for the weight functions in Appendix \ref{sec:weFunction}.

\subsubsection{Selection of transition point \texorpdfstring{$s$}{s}}
\label{sec:s}

The DIG algorithm operates in two phases: an initial exploration phase where we aggressively reallocate poorly-assigned observations, followed by a convergence phase where adaptation diminishes. The transition point $s$ determines when to switch between these phases. Ideally, most observations should be correctly allocated by iteration $s$.

Consider a simplified algorithm which selects at each iteration 1 observation uniformly at random, and suppose that each observation is allocated to the correct mixture component according to an independent Bernoulli random variable with probability $1/K$ (uniform over the number of components). The number of iterations $S$ needed to allocate all observations to their respective component can be approximated by a negative binomial distribution with parameters $r = n$ and $p = 1/K$, which has mean $\mu = n(K-1)$ and variance $\sigma^{2} = nK(K-1)$. Since DIG selects $m$ observations at each iteration, we then set $s$ such that $ms = \mu + \sigma Z_{0.99}$, i.e.\ the 99\% quantile of the normal approximation to the confidence interval of a negative binomial, to give the chain a little more time to adjust during the greedy adaptation phase of the algorithm. Thus,
\begin{equation} \label{eq:s}
s = \frac{n(K-1) + \sqrt{nK(K-1)} \cdot Z_{0.99}}{m} \, .
\end{equation}

Note that while this model assumes independent allocation attempts with uniform probabilities and the existence of a unique correct allocation for each observation, the actual DIG algorithm uses adaptive probabilities and may encounter cases where allocation is genuinely ambiguous due to overlapping clusters or model misspecification. Nevertheless, this approach provides effective guidance for setting $s$, as demonstrated in our numerical experiments (Section~\ref{sec:numExp}).

\subsection{Stochastic approximation analysis} \label{sec:stoc}

After iteration $s$, when $\lambda_t = 1$ and the polynomial weights take effect, our adaptive scheme fits naturally into the stochastic approximation framework. This allows us to establish convergence of the selection probabilities to a well-defined limit.

The polynomial weights from Equation~\eqref{eq:wFunPol} can be rewritten in the canonical stochastic approximation form:
\begin{equation*}
    \boldsymbol{\alpha}_t = \boldsymbol{\alpha}_{t-1} + \frac{1}{t-s+2} \left( \boldsymbol{D}_t^{1} - \boldsymbol{\alpha}_{t-1}  \right),
\end{equation*} 
which has step size $1/(t-s+2)$ and ``target'' function $\boldsymbol{D}_t^{1}$.

\begin{theorem}\label{thm:SA_conv}
For finite Gaussian mixtures under our adaptive scheme, the selection probabilities converge almost surely:
\begin{equation}
    \boldsymbol{\alpha}_t \to \mathbb E_{\nu^{\ast}}[\boldsymbol{D}^1] \quad \text{as } t\to \infty \, ,
    \label{eq:conv}
\end{equation}
where $\nu^{\ast}$ denotes the target posterior distribution and $\boldsymbol{D}^1$ is the discomfort vector in \eqref{DivFun} with $\lambda = 1$.
\end{theorem}

\noindent \textbf{Proof:} The result follows from the stochastic approximation framework of \citet{Andrieu2005}. The polynomial step sizes satisfy $\sum \rho_t = \infty$ and $\sum \rho_t^2 < \infty$, the target function $h(\boldsymbol{\alpha}) = \mathbb E_{\nu^{\ast}}[\boldsymbol{D}^1] - \boldsymbol{\alpha}$ has unique stationary point $\mathbb E_{\nu^{\ast}}[\boldsymbol{D}^1]$, and the noise terms satisfy appropriate martingale conditions. Full details are provided in Appendix~\ref{app:SA}.

This result establishes that the adaptive weights converge to the expected discomfort under the target distribution---precisely the quantity we seek to approximate. Observations that remain poorly assigned under the posterior will maintain high limiting weights, whilst well-assigned observations will have low limiting weights.  Empirical evidence for this convergence is provided in Figure \ref{fig:resultConv} in Appendix~\ref{app:SA}.

\subsection{Complete algorithm} \label{sec:completeAlgo}

The complete DIG sampler which integrates all the adaptive components described above is provided in  Appendix~\ref{completeAlgorithm} (Algorithm \ref{alg:ARSG4}).

We note that the DIG sampler requires $\mathcal{O}(mK)$ likelihood evaluations per iteration (same as RSG), plus periodic $\mathcal{O}(nK)$ evaluations for updating the allocation matrix. The root-finding for $\lambda_t$ adds negligible overhead in practice. Overall computational cost remains substantially lower than SSG's $\mathcal{O}(nK)$ per iteration when $m \ll n$.

\section{Simulation study} \label{sec:numExp}

In this section, we empirically investigate our DIG algorithm and compare its convergence rate using the CLL as a performance metric. Section \ref{Hyper} details and justifies the choice of hyperparameters for the proposed model. Section \ref{simStud} presents results across different scenarios. In Section \ref{over}, we consider the case of an overestimated number of clusters, while Section \ref{modMiss} reports results under model misspecification. In particular, we benchmark our model against established alternatives, including SSG and RSG.

\subsection{Hyperparameters and free parameters settings} \label{Hyper} 

The update size $m$ is set to $1\%$, $2\%$, and $3\%$ of the total number of observations $n$ for $n = 1000$, $n = 5000$, and $n = 10{,}000$, respectively. This increasing proportion is chosen to ensure that, as the dataset grows, the sampler retains sufficient information at each iteration to maintain adequate mixing and approximation quality. For smaller datasets, a larger $m$ would provide limited computational gain, while for larger $n$, updating a very small subset would result in slow convergence and reduced performance compared to SSG method. The allocation probability matrix is updated following the schedule described in Section \ref{allUp}. During simulations, the decay parameter $\lambda$ is dynamically adjusted via Equation \eqref{eq:lambda_schedule} concurrently with each matrix update, subject to a maximum $\Lambda = 100$ as established in Section \ref{sec:disDec}. As highlighted in Section \ref{sec:theory}, we employ two distinct pairs of functions before and after the intersection point $s$, which we define in Equation \eqref{eq:s}. Specifically, before $s$, we apply the hyperbolic tangent functions, as in Equation \eqref{eq:wFunHyper}, with $a = 1$. After $s$, we switch to the polynomial pair of functions described in Equation \eqref{eq:wFunPol}.
We adopt the same prior specification and hyperparameters, as in Section \ref{sec:motivating}.
For more details on hyperparameter settings and input parameter choices, see Appendix \ref{app:free}.

\subsection{Gaussian mixtures} 
\label{simStud}

We follow the data-generating mechanism proposed by \citet{miller2018}, i.e.\ a Gaussian mixture with $K^*=K=3$, with centres $(-3/\sqrt{d}, \; 0, \; 3/\sqrt{d})$, spherical covariance, and a varying number of observations and dimensions. 
We run DIG, SSG, and RSG for 20 independent replicas, without applying burn-in or thinning. 
In order to control for randomness, we initialise each replica of a single algorithm with a different random seed. The same 20 seeds are then re-utilised for the 20 replicas of every algorithm studied.

To evaluate the time to convergence, we adopted an approach inspired by the Geweke diagnostic \citep{geweke1991evaluating} that uses the SSG sampler as a reference standard. Specifically, we computed a reference value from the 20 SSG chains by taking the mean CLL over the final 1000 iterations of each chain, then averaging these means across all 20 chains. For each RSG and DIG chain, we used a sliding window of 1000 iterations and calculated the mean CLL within each window. We computed a z-score comparing each window's mean to the SSG reference value, accounting for their respective variances. We declared convergence when $|z| < 1.96$, indicating failure to reject equality at the 95\% confidence level.

Table \ref{resultsARI} in the appendix shows the CLL value for the experiments conducted in this section. All the values are very similar, suggesting that all the algorithms converge to the same configuration.

\begin{table}[htbp]
\centering
\renewcommand{\arraystretch}{1.2}
\begin{tabular}{@{}clcccc@{}}
\toprule
& & \multicolumn{4}{c}{\textbf{Mean Time to Converge (s)}} \\
\cmidrule(lr){3-6}
\textbf{Sample Size} & \textbf{Method} & $d = 2$ & $d = 5$ & $d = 10$ & $d = 20$ \\
\midrule
\multirow{3}{*}{$n = 1000$} 
& SSG & 1.703 (1.318) & 2.528 (2.879) & 2.894 (3.297) & 1.107 (1.984) \\
& RSG & 1.187 (0.434) & 1.197 (0.522) & 1.315 (0.625) & 2.748 (1.295) \\
& \textbf{DIG} & \textbf{0.005} (0.005) & \textbf{0.015} (0.020) & \textbf{0.023} (0.007) & \textbf{0.038} (0.028) \\
\midrule
\multirow{3}{*}{$n = 5000$}
& SSG & 4.566 (4.796) & 5.051 (6.194) & 10.075 (20.509) & 4.977 (10.201) \\
& RSG & 2.769 (0.866) & 2.996 (1.909) & 4.176 (3.306) & 10.914 (8.446) \\
& \textbf{DIG} & \textbf{0.530} (0.657) & \textbf{0.945} (2.242) & \textbf{1.159} (2.070) & \textbf{1.043} (1.775) \\
\midrule
\multirow{3}{*}{$n = 10000$}
& SSG & 12.815 (15.669) & 30.276 (27.168) & 41.309 (61.458) & 10.475 (18.722) \\
& RSG & 5.686 (2.694) & 6.210 (4.961) & 6.565 (5.200) & 12.511 (9.357) \\
& \textbf{DIG} & \textbf{1.298} (1.431) & \textbf{0.791} (0.811) & \textbf{0.758} (0.675) & \textbf{1.970} (2.698) \\
\bottomrule
\end{tabular}
\caption{Miller-Harrison synthetic datasets across different scenarios: Mean Time to Convergence (in seconds) over 20 replicas (with standard errors reported in parentheses) for the different scenarios across the SSG, RSG, and DIG algorithms. In each scenario, the best performance is highlighted in bold.}
\label{resultsTime}
\end{table}

In Table \ref{resultsTime} we report the average time to convergence (in seconds) over 20 replicas for the SSG, RSG, and DIG algorithms, across different values of sample size $n$ and dimensionality $d$. 
Overall, DIG shows the most robust performance, achieving faster convergence in all settings (by at least one order of magnitude). Its efficiency is particularly noticeable as the dimensionality increases, where standard methods like SSG—which updates all allocations at each iteration—suffer from significant slowdowns. RSG, which randomly updates a subset of observations, offers an intermediate trade-off: it improves upon SSG in several cases but remains sensitive to the choice of $m$ and shows more variability.

\begin{figure}[htbp]
    \centering
\begin{subfigure}[c]{0.47\textwidth}  
    \includegraphics[width=\textwidth]{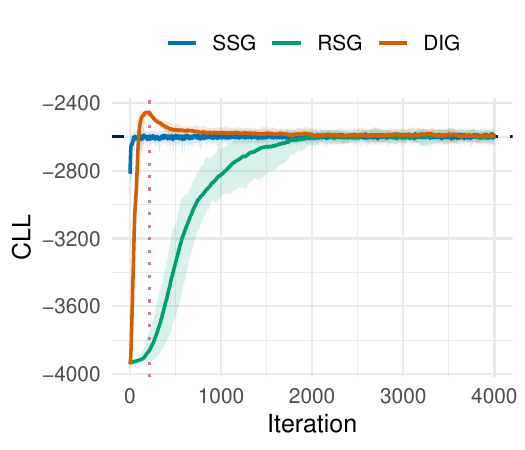}
    \caption{Complete Log-Likelihood (CLL) of SSG, RSG and DIG. Solid lines indicate the mean across 20 independent replicas, with shaded regions representing the 95\% credible interval. The dashed black line indicates the average computed over the last 1000 iterations of the SSG. The dotted vertical line marks the intersection point $s$.}
    \label{fig:CLLScen}
\end{subfigure}
\hspace{0.04\textwidth}  
\begin{subfigure}[c]{0.47\textwidth}  
    \includegraphics[width=\textwidth]{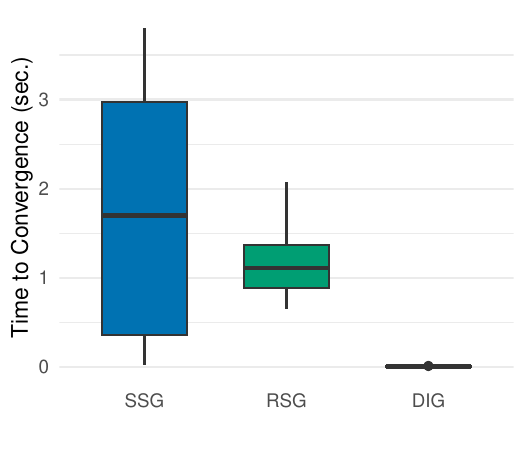}
    \caption{Box plots of the time to convergence (in seconds) for SSG, RSG, and DIG across 20 independent replicas.}
    \label{fig:TimeScen}
\end{subfigure}
    \caption{Miller-Harrison synthetic dataset with $n=1000$, $d=2$ and $K=3$. 
    }
    \label{fig:resultToy}
\end{figure}

Figure \ref{fig:CLLScen} illustrates the results for the CLL and convergence time in the scenario where $n = 1000$ and $d = 2$. The left panel shows that SSG and DIG converge very quickly compared to RSG. 
Interestingly, DIG initially reaches a higher CLL before converging to the correct level.  We note that this behaviour was not observed in our motivating example based on \citet{miller2018} in Section \ref{sec:motivating} in which the clusters were well separated and the correct allocations had little uncertainty associated with them. In contrast, in the present example, the degree of overlap between clusters is much greater, which leads to higher uncertainty in allocations, and decreases the sojourn time of the chain at the true allocation. However, during its initial greedy phase, an adaptive algorithm like DIG will often find allocations that have higher CLL than the configurations in the typical set of the chain, and will then visit these configurations less frequently once the adaptation weights start settling.
Figure \ref{fig:TimeScen} shows the wall clock convergence time of the algorithms, measured in seconds, as this accounts for the difference in computational cost. We observe that DIG achieves the fastest convergence time compared to its competitors, followed by RSG and then SSG.

It is worth noting that the SSG algorithm updates all allocation variables at each iteration, thereby exploring the entire dataset simultaneously. In contrast, RSG and DIG update only a subset of $m$ observations per iteration, requiring approximately $n/m$ iterations to cover all data points. For this reason, in Figure \ref{fig:resultToyEpoch}, we report the CLL as a function of the number of epochs---where one epoch corresponds to $n/m$ iterations--for the case $n = 1000$, $d = 2$. This representation enables a more balanced comparison among the three methods by normalizing the total number of updates performed.

\begin{figure}[htbp]
    \centering
    \includegraphics[scale = 0.63]{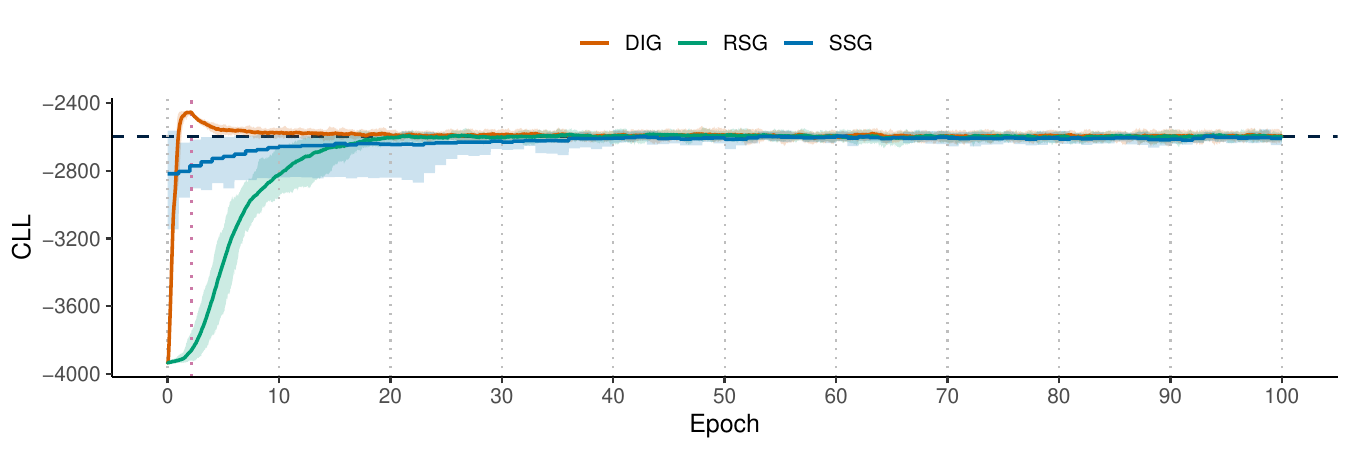}
    \caption{Complete Log-Likelihood (CLL) for the dataset proposed by \citet{miller2018}, with $n = 1000$ and $d = 2$, comparing SSG, RSG, and DIG. Solid lines indicate the mean across 20 independent replicas, with shaded regions representing the 95\% credible interval. The dashed black line indicates the average computed over the last 1000 iterations of the SSG. The vertical lines indicate epochs, where a single iteration of SSG corresponds to $n/m$ iterations of RSG and DIG. The dashed red vertical line represent the intersection point $s$.} 
    \label{fig:resultToyEpoch}
\end{figure}

\noindent Both RSG and DIG exhibit convergence in terms of the CLL, using the SSG method as a reference, across all considered scenarios. This behaviour is consistent as the sample size increases and the dimensionality becomes more complex, confirming the stability of the methods. Regarding the time required to reach convergence, and again taking SSG as a benchmark, the DIG method shows significantly better computational performance compared to the other two approaches. This trend remains evident across different data dimensions and sample sizes.

\subsection{Overfitted Gaussian mixtures} 
\label{over}

Using the data generated by the mechanism proposed in \citet{miller2018}, we set the number of components $K^{\ast}$, to be greater than the number of clusters, $K$. We compare the DIG algorithm with the two previously discussed algorithms by generating $n = 1000$ observations with $d = 2$, running the MCMC for 20 independent replicas. The remaining input parameters and hyperparameters are consistent with those reported in Section \ref{Hyper}. Recalling that in this case the number of clusters is three, i.e.\ $K = 3$, we show in Table \ref{resMiss} the results in terms of CLL and time to convergence for $K^{\ast} = \{5, 20\}$.

\begin{table}[htbp]
\renewcommand{\arraystretch}{1.4}
\centering
\begin{tabular}{clcc}
\hline
       & \multicolumn{1}{c}{} & CLL & T2C \\ \cline{3-4} 
$K^{\ast} = 5$  & SSG                  & \textbf{-2601.592} (6.086)    &  0.337 (0.393)   \\
       & RSG                  & -2620.822 (4.076)    & 14.230 (5.075)    \\
       & DIG             & -2603.689 (2.457)    & \textbf{0.115} (0.036)    \\ \hline
$K^{\ast}= 20$ & SSG                  & -2610.879 (8.251)    & 1.651 (1.820)    \\
       & RSG                  & -2759.763 (5.070)    & 62.027 (11.749)    \\
       & DIG             & \textbf{-2603.820} (5.205)    & \textbf{1.090} (0.123)   \\ \hline
       \end{tabular}
    \caption{Performance with overfitted models ($K^{\ast} > K$, $n = 1000$, and $d = 2$): Mean Complete Log-Likelihood (CLL) and Mean Time to Convergence (T2C, in seconds) over 20 replications, with standard errors in parentheses, for the Zeisel dataset across the SSG, RSG, and DIG algorithms. The best performance is highlighted in bold.}
    \label{resMiss}
\end{table}

\noindent Table \ref{resMiss} reports the values of the complete log-likelihood (CLL) and the average time to convergence for the three algorithms, evaluated across different values of $K^\ast$. The CLL results are very close for all methods, indicating that each algorithm is capable of reaching similar posterior regions.

Looking at the convergence time, DIG shows consistently faster performance (by at least one order of magnitude) in both scenarios, suggesting that its informed updates help the chain mix more efficiently. SSG also shows good performance, especially when the model complexity is lower. Conversely, RSG generally requires more time, and this is particularly evident when $K^\ast = 20$, likely because it lacks a mechanism to prioritize more informative updates.

Overall, DIG seems to strike a good balance between convergence speed and estimation quality, while SSG remains a solid and reliable baseline, especially in less complex settings.

\begin{figure}[H]
    \centering
    \includegraphics[scale = 0.54]{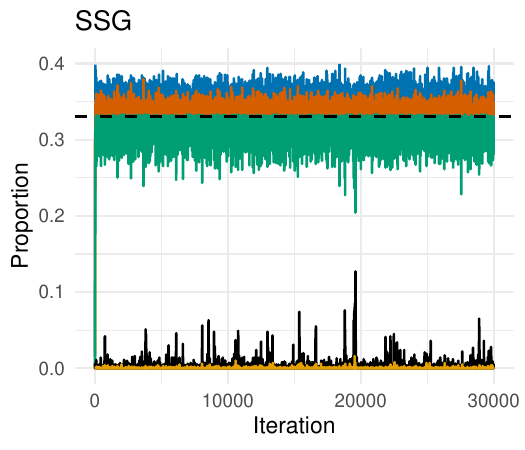}
    \includegraphics[scale = 0.54]{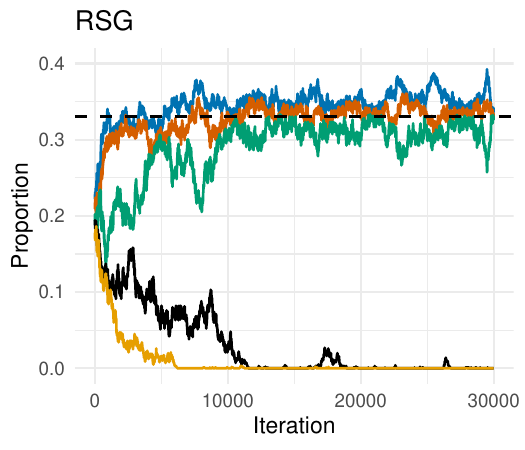}
    \includegraphics[scale = 0.54]{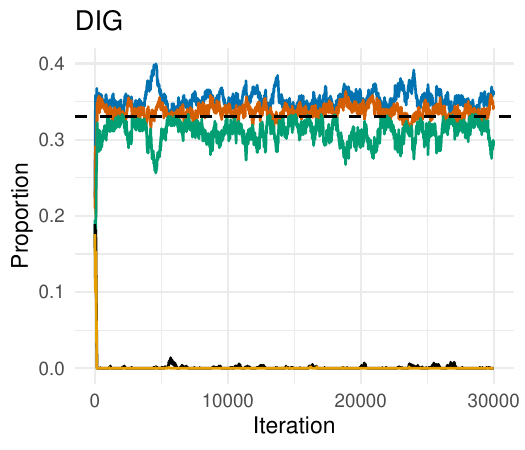}
    \includegraphics[scale = 0.54]{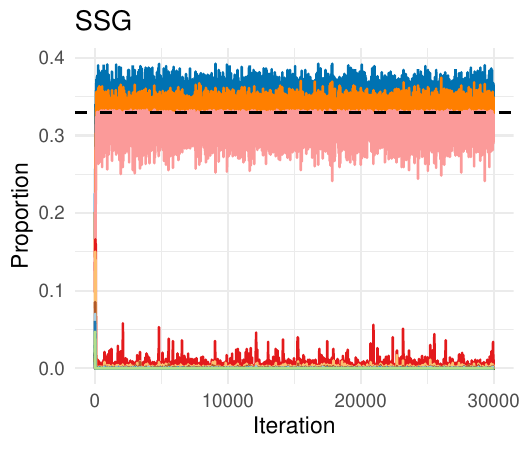}
    \includegraphics[scale = 0.54]{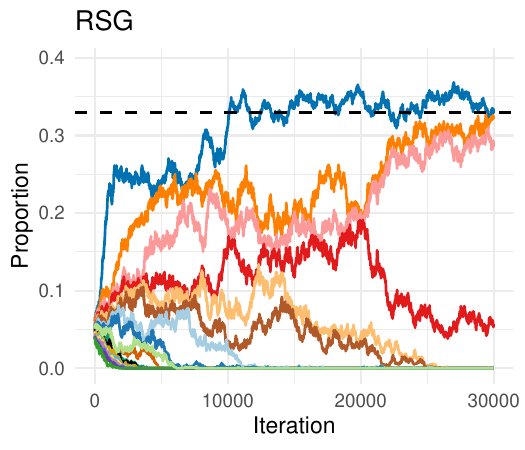}
    \includegraphics[scale = 0.54]{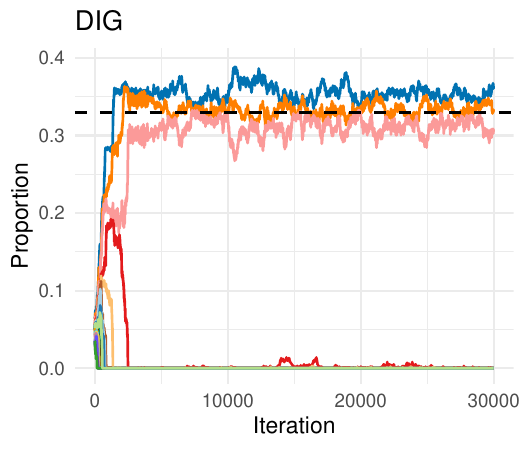}
    \caption{Component proportion over iterations for one MCMC chain of SSG, RSG, and DIG on the Miller-Harrison dataset ($K = 3$). \textit{Upper row}: $K^{\ast} = 5$; \textit{Lower row}: $K^{\ast} = 20$. Each colored line represents a distinct estimated component, while the black dotted line indicates the true cluster proportions.} 
    \label{fig:proportionPlot}
\end{figure}

\noindent Figure \ref{fig:proportionPlot} shows the proportion of component occupancy over the course of the iterations. The plots in the first row correspond to the three methods: SSG, RSG, and DIG, respectively, with $K^{\ast} = 5$, while those in the second row correspond to $K^{\ast} = 20$. We observe that SSG and DIG are able to identify the three main clusters early on, setting the proportions of the remaining components to zero within the initial iterations. In contrast, RSG requires more iterations to recognise the presence of only three clusters.

\subsection{Model misspecification} \label{modMiss}

Our model, as well as the SSG and RSG methods, assumes a diagonal variance-covariance matrix. To assess the robustness of these approaches, we generate synthetic data using a finite mixture model inspired by \citet{miller2018} with four clusters; see Figure \ref{fig:Misspecification}. The synthetic bivariate observations are distributed according to:
\begin{align*}
z_i \mid \bm \pi &\sim \text{Categorical}(\bm \pi),\notag \quad i = 1, \ldots, n, \\
x_i &\sim \mathcal{N}(\mu_{z_i}, \Sigma_{z_i}) \notag \quad i = 1, \ldots, n, 
\end{align*}
where
\begin{align*}
& \bm \pi  = (0.44, \, 0.3, \, 0.25, \, 0.01), \quad
\mu_1 = (4, \, 4)^{\top} , \mu_2 = (7, \, 4)^{\top}, \mu_3 = (6, \, 2)^{\top}, \mu_4 = (8, \, 10)^{\top} \, , \\
& \Sigma_1 = \begin{pmatrix}
1 & 0 \\
0 & 1
\end{pmatrix}, \, 
\Sigma_2 = R
\begin{pmatrix}
2.5 & 0 \\
0 & 0.2
\end{pmatrix} 
R^{\top}, \,
\Sigma_3 = \begin{pmatrix}
3 & 0 \\
0 & 0.1
\end{pmatrix}, \,
\Sigma_4 = 
\begin{pmatrix}
0.1 & 0 \\
0 & 0.1
\end{pmatrix} \, ,
\end{align*}

with 
\begin{equation*}
R =  \begin{pmatrix}
\cos(\pi/4) & - \sin(\pi/4) \\
\sin(\pi/4) & \cos(\pi/4)
\end{pmatrix} \, .
\end{equation*}

\begin{figure}[!ht]
    \centering
    \includegraphics[scale = 1]{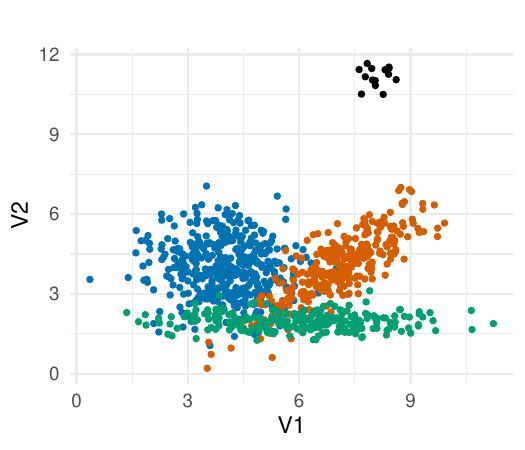}
    \caption{Synthetic Gaussian data with four clusters ($K = 4$), $n = 1000$ observations and $d = 2$ for model misspecification testing.} 
    \label{fig:Misspecification}
\end{figure}

\begin{figure}[htbp]
    \centering
    \begin{subfigure}[c]{0.48\textwidth}
        \includegraphics[width=\textwidth]{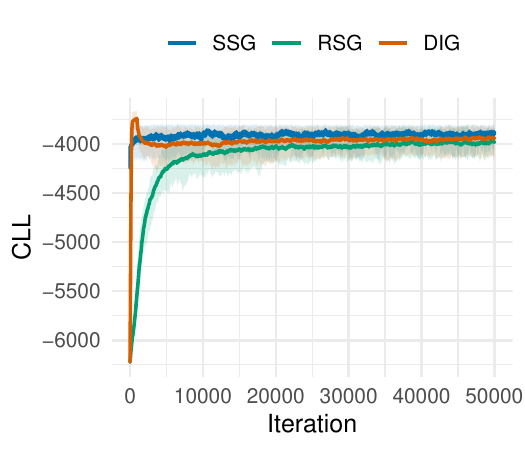}
        \caption{Complete Log-Likelihood (CLL) of SSG, RSG and DIG. Solid lines indicate the mean across 20 independent replicas, with shaded regions representing the 95\% credible interval.}
        \label{fig:misspec_cll}
    \end{subfigure}
    \hfill
    \begin{subfigure}[c]{0.48\textwidth}
        \includegraphics[width=\textwidth]{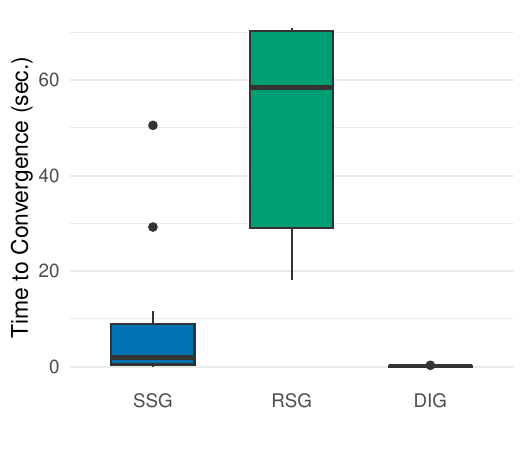}
        \caption{Box plots of the time to convergence (in seconds) for SSG, RSG, and DIG across 20 independent replicas}
        \label{fig:misspec_time}
    \end{subfigure}
    \caption{Model performance under misspecification using synthetic 4-cluster dataset with n = 1000 and d = 2. 
    }
    \label{fig:MisspecificationResult}
\end{figure}

\begin{figure}[!ht]
    \centering
    \includegraphics[scale = 0.54]{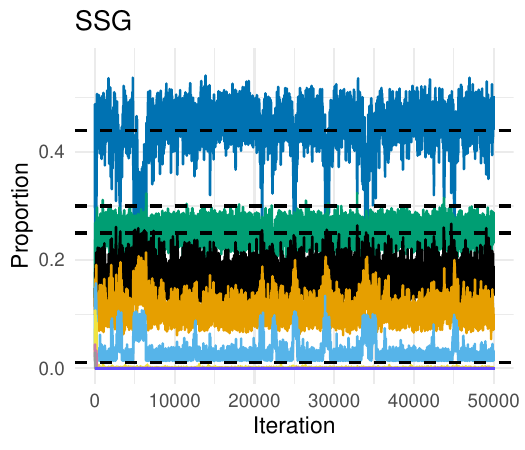}
    \includegraphics[scale = 0.54]{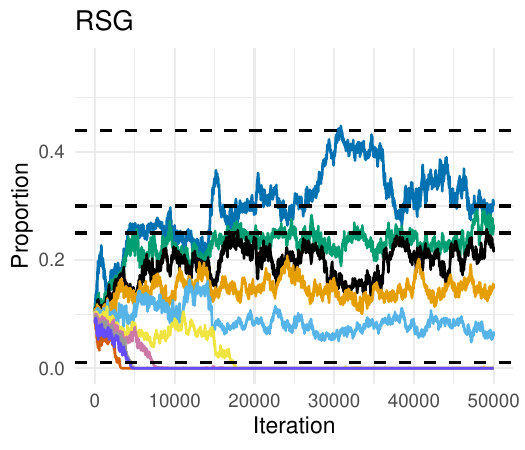}
    \includegraphics[scale = 0.54]{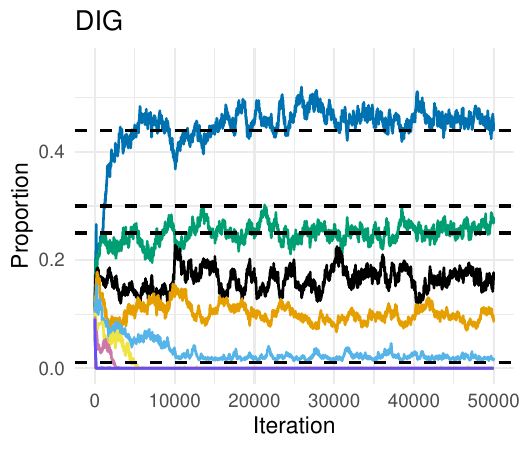}
    \caption{Component evolution for the misspecified model on the 4-cluster dataset, fitted with $K^{\ast} = 10$. The plot shows component occupancy proportions over iterations for one chain of SSG, RSG, and DIG. Each colored line represents a distinct estimated component, while the black dotted line indicates the true cluster proportions. Since we estimate the model using diagonal covariance matrices, we are not able to fully capture the true covariance structure of the data. As a result, we expect to observe a greater number of active components than the true number of clusters. In this case, we note that all algorithms consistently identify 5 clusters.} 
    \label{fig:MisspecificationProp}
\end{figure}

\noindent In Table \ref{tab:resMiss} and Figure \ref{fig:MisspecificationResult}, we observe that, while the DIG method achieves similar CLL values to the other two approaches, it stands out for its significantly faster convergence time. The SSG method shows comparable computational efficiency, whereas RSG reaches convergence in terms of CLL much later. \\
Furthermore, given that the data are modelled with a total of 10 components, i.e.\ $K^{\ast} = 10$, Figure \ref{fig:MisspecificationProp} illustrates how all methods identify the presence of 5 distinct components. Although the true data-generating process involves a mixture of 4 components, this discrepancy can be considered reasonable given that the models do not employ a full covariance matrix.

\begin{table}[htbp]
\renewcommand{\arraystretch}{1.4}
\centering
\begin{tabular}{clcc}
\hline
       &  & CLL & T2C \\ \cline{3-4}   & SSG                  & \textbf{-3894.962} (7.512)    &  7.001 (12.344)   \\
       & RSG                  & -3963.987 (5.161)    & 51.356 (20.900)    \\
       & DIG             & -3940.558 (3.872)    & \textbf{0.090} (0.067)    \\ \hline
       \end{tabular}
    \caption{Performance under model misspecification: Mean Complete Log-Likelihood (CLL) and Mean Time to Convergence (T2C, in seconds) over 20 replications, with standard errors in parentheses, for the Zeisel dataset across the SSG, RSG, and DIG algorithms. The best performance is highlighted in bold.}
    \label{tab:resMiss}
\end{table}

\section{Real world applications} \label{sec:real}

In this section, we illustrate the practical utility of the proposed approach through the analysis of two real datasets, in comparison with well-established and alternative approaches. We compare the proposed model with SSG and RSG in terms of the time to convergence measured in the CLL, and, to evaluate clustering performance, we also report the Adjusted Rand Index (ARI) \citep{hubert1985comparing}, since ground-truth labels are available for both datasets. The ARI is a measure used to evaluate the similarity between two data clusterings defined as follows:
\begin{equation*}
    ARI = \frac{\sum_{ij} \binom{n_{ij}}{2} - \left[ \sum_i \binom{a_i}{2} \sum_j \binom{b_j}{2} \right] / \binom{n}{2}}{ \frac{1}{2} \left[ \sum_i \binom{a_i}{2} + \sum_j \binom{b_j}{2} \right] - \left[ \sum_i \binom{a_i}{2} \sum_j \binom{b_j}{2} \right] / \binom{n}{2}} \, ,
\end{equation*}
where, $n_{ij}$ is the number of observations in cluster $i$ of the first clustering and in cluster $j$ of the second clustering, $a_i$ is the number of observations in cluster $i$ of the first clustering, and $b_j$ is the number of observations in cluster $j$ of the second clustering. For consistency, we employ the same prior specification and hyperparameter settings as those used in the simulation studies for the proposed model.
The purpose of the applications presented in this section is to demonstrate how the DIG model can be applied to a wide range of real-world scenarios.
As an overall comment, the proposed DIG demonstrates better performance compared to the other methods. In particular, DIG shows good performance in the two datasets, with low $n$ and low $d$ and high $n$ and high $d$.

\subsection{Single-Cell RNA-Seq Data from Somatosensory Cortex Cells} \label{sec:real-soma}

We analyse gene expression data from single-cell RNA sequencing of mouse somatosensory cortex and hippocampus cells as reported by \citet{zeisel2015cell}. While the original study identified 47 distinct cell subclasses across 3,005 cells and 19,972 genes, we follow \citet{prabhakaran2016dirichlet} in aggregating these into 7 major cell type categories for clustering evaluation: (1) somatosensory pyramidal neurons, (2) CA1 pyramidal neurons, (3) interneurons, (4) oligodendrocytes, (5) astrocytes, (6) microglia, and (7) endothelial cells. This aggregation combines biologically related subclasses (e.g., the 16 interneuron subtypes identified by Zeisel {\em et al.} are treated as a single interneuron category) to create ground truth cluster labels that capture major cell type distinctions.  We preprocess the data following the approach of \citet{prabhakaran2016dirichlet}, but as in \citet{chaumeny2022bayesian}, we limit the analysis to the $d = 10$ genes with the highest standard deviations. We model the data assuming $K^{\ast} = K=7$, corresponding to these major cell type categories, and adopt a spherical covariance structure.

\begin{table}[htbp]
\renewcommand{\arraystretch}{1.4}
\centering
\begin{tabular}{clccc}
\hline
       &  & CLL & ARI & T2C \\ \cline{3-5}
       & SSG  & -47848.129 (8.907) &  0.305 (0.001) & 621.680 (769.366) \\
       & RSG  & \textbf{-47584.913} (7.1209) &  0.330 (0.001) & 178.883 (119.127) \\
       & DIG  & -47682.224 (7.7999) &  \textbf{0.343} (0.001) & \textbf{87.881} (168.306) \\ \hline
\end{tabular}
\caption{Performance on Mouse Cortex scRNA-seq dataset: Mean Complete Log-Likelihood (CLL), Adjusted Rand Index (ARI), and Mean Time to Convergence (T2C, in seconds) over 20 replications, with standard errors in parentheses, across the SSG, RSG, and DIG algorithms. The best performance is highlighted in bold.}
\label{resZei}
\end{table}

\noindent In Table \ref{resZei}, we report the results for the CLL and the time to convergence. On average, the three methods achieve comparable CLL values, indicating similar clustering performance. However, the time required to reach convergence is substantially shorter for the DIG method compared to both SSG and RSG.
Figure \ref{fig:ZeiData} further illustrates the behaviour of the three algorithms in terms of CLL evolution and convergence times. Notably, some replications of the SSG method fail to achieve convergence within the allocated iterations, while all replications of RSG and DIG successfully converge. This highlights a potential limitation of SSG in this scenario, and reinforces the robustness and efficiency of DIG under the same conditions.

\begin{figure}[htbp]
    \centering
    \begin{subfigure}[c]{0.48\textwidth}
        \includegraphics[width=\textwidth]{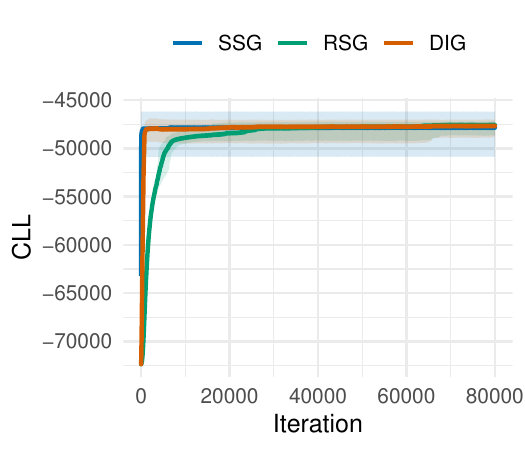}
        \caption{Complete Log-Likelihood (CLL) of SSG, RSG and DIG. Solid lines indicate the mean across 20 independent replicas, with shaded regions representing the 95\% credible interval. Some SSG chains do not converge.}
        \label{fig:gene_cll}
    \end{subfigure}
    \hfill
    \begin{subfigure}[c]{0.48\textwidth}
        \includegraphics[width=\textwidth]{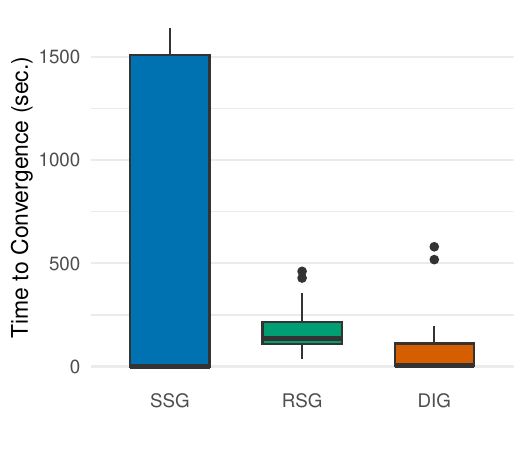}
        \caption{Boxplots of the time to convergence (in seconds) for SSG, RSG, and DIG across 20 independent replicas.}
        \label{fig:gene_time}
    \end{subfigure}
    \caption{Mouse Cortex scRNA-seq dataset performance (\citet{zeisel2015cell}: n=3005, d=10). 
    }
    \label{fig:ZeiData}
\end{figure}

In Figure \ref{fig:matrixGENE}, Appendix \ref{app:heatmaps}, we show the Posterior Similarity Matrix (PSM) for a single MCMC chain (upper panel) and a heatmap comparing the clusters recovered by the DIG algorithm with the true labels (lower panel).
The Posterior Similarity Matrix (PSM) reveals a well-defined block-diagonal pattern, indicating the presence of 7 clusters as inferred by the model. Among these, two clusters are relatively large, while the remaining five are smaller in size. Since the number of components was fixed to 7 in the estimation procedure, this result is expected. However, the sharpness of the blocks indicates that the sampler provides stable assignments within clusters across iterations.

\subsection{Breast Cancer Molecular Subtypes via the PAM50 Gene Signature} \label{sec:real-pam50}

The PAM50 classifier, introduced by \citet{parker2009supervised}, uses expression levels of 50 genes to categorise breast tumours into intrinsic molecular subtypes:
\emph{Luminal A}, \emph{Luminal B}, \emph{HER2-enriched}, \emph{Basal-like}, and \emph{Normal-like}.
This subtype classification consistently proves to be a strong, independent predictor of patient survival, even when considered together with all the standard clinical and pathological factors \citep{parker2009supervised, hu2006molecular, nielsen2010comparison, cheang2012responsiveness, dowsett2013comparison, gnant2014predicting, prat2015clinical}.
The PAM50 gene set has been widely adopted in research and clinical practice, mainly because it demonstrates strong concordance with the much larger and more complex intrinsic gene sets \citep{perou2000molecular, sorlie2003repeated, nielsen2014analytical}.
The aim of our analysis is to identify clusters that, based on the characteristics of the observations, correspond to the molecular subtypes in the dataset under examination. In particular, we analyse a breast cancer gene expression dataset from the \citet{cancer2012tcga}, processed as in \citet{Lock2013} This encompasses $n = 348$ samples and the $d = 50$ PAM50 genes, providing a comprehensive foundation for understanding these molecular distinctions. We model the data assuming $K^{\ast} = K=5$, and a spherical covariance structure.

\begin{table}[htbp]
\renewcommand{\arraystretch}{1.4}
\centering
\begin{tabular}{clccc}
\hline
       &  & CLL & ARI & T2C \\ \cline{3-5}
       & SSG  & \textbf{-17943.844} (4.123) &  0.426 (0.002) & 7.611 (11.336) \\
       & RSG  & -17983.424 (4.117) &  0.452 (0.002) & 6.351 (4.267) \\
       & DIG  & -17957.468 (3.985) &  \textbf{0.453} (0.002) & \textbf{2.912} (5.541) \\ \hline
\end{tabular}
\caption{Performance on PAM50 breast cancer data: Mean Complete Log-Likelihood (CLL), Adjusted Rand Index (ARI), and Mean Time to Convergence (T2C, in seconds) over 20 replications, with standard errors in parentheses, across the SSG, RSG, and DIG algorithms. The best performance is highlighted in bold.}
\label{resPAM}
\end{table}

Table \ref{resPAM} reports the results obtained by the three methods in terms of CLL and time to convergence, based on 20 replications. The results show that the DIG method converges significantly faster than both SSG and RSG, reducing the average convergence time by about half. This computational advantage does not come at the cost of solution quality: all three methods achieve very similar CLL values, indicating that they explore the posterior distribution effectively. Figure \ref{fig:ariPAM} displays the evolution of the CLL over time, along with the observed convergence times for each method, highlighting the differences in efficiency.

\begin{figure}[htbp]
    \centering
    \begin{subfigure}[c]{0.48\textwidth}
        \includegraphics[width=\textwidth]{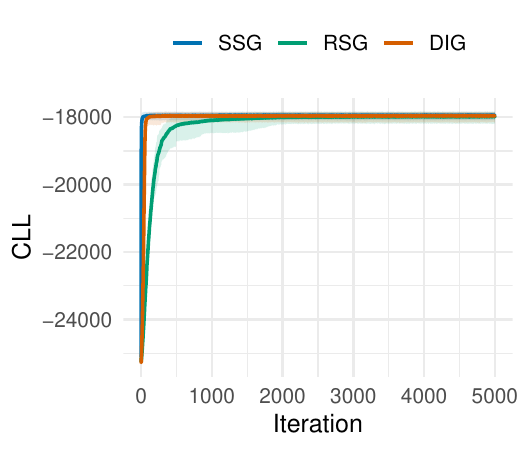}
        \caption{Complete Log-Likelihood (CLL) of SSG, RSG and DIG. Solid lines indicate the mean across 20 independent replicas, with shaded regions representing the 95\% credible interval. Some chains of SSG does not converge.}
        \label{fig:pam50_cll}
    \end{subfigure}
    \hfill
    \begin{subfigure}[c]{0.48\textwidth}
        \includegraphics[width=\textwidth]{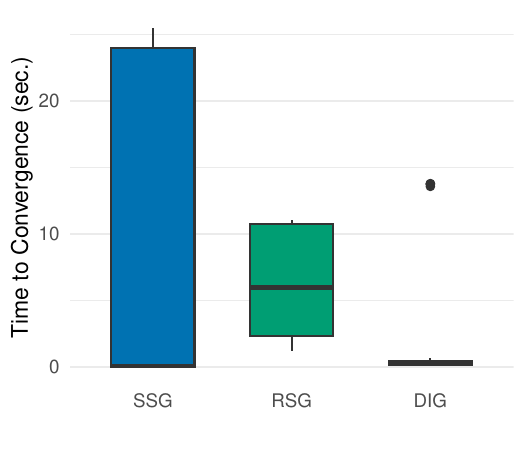}
        \caption{Box plots of the time to convergence (in seconds) for SSG, RSG, and DIG across 20 independent replicas.}
        \label{fig:pam50_time}
    \end{subfigure}
    \caption{PAM50 breast cancer dataset performance (n=348 samples, d=50 genes). 
    }
    \label{fig:ariPAM}
\end{figure}

In Figure \ref{fig:PAMData}, Appendix \ref{app:heatmaps}, we show the PSM from a single MCMC chain (upper panel) alongside a heatmap comparing the clusters identified by the DIG algorithm with the true labels (lower panel). The results show that the method successfully detects 5 clusters, consisting of 4 large clusters and 1 very small cluster.

\section{Conclusion} \label{sec:disc}

We have proposed an adaptive Gibbs sampling algorithm for finite mixture models, designed to improve convergence efficiency and computational tractability. Unlike standard approaches that rely on uniform selection of observations, our method dynamically adjusts the selection mechanism to prioritise observations that are more uncertain with respect to their current component allocation. This adaptive strategy increases the sampler’s ability to explore the posterior distribution efficiently, particularly during the initial convergence phase of the algorithm.

The selection mechanism evolves over time through a weighting scheme that balances past selection probabilities with information about current allocation uncertainty. Early iterations emphasise poorly allocated observations, while later updates become more stable as the model reaches a more reliable configuration. This dynamic adjustment is controlled by a decay parameter that gradually reduces the impact of the adaptive component, ensuring long-term stability.

We showed that the adaptive weights remain valid throughout the sampling process, maintaining positive probabilities for all observations. This property ensures that the Markov chain preserves desirable convergence properties and satisfies the theoretical conditions required for ergodicity.

Through simulation studies and data applications, we demonstrated that the proposed method achieves improved performance in terms of convergence speed and clustering accuracy when compared to standard approaches. In particular, the adaptive mechanism allows the algorithm to focus computational effort where it is most needed, leading to a more efficient use of sampling iterations. Notably, our adaptive approach is essentially automatic, it does not require any user input, and it does not require any extra calculations once one has derived the required conditional distributions for a Systematic Gibbs Sampler.

Although our focus has been on finite mixture models, the general idea of adaptively modulating selection probabilities based on model uncertainty can be extended to other classes of latent variable models. Future developments may explore more general adaptation strategies, alternative decay schedules, and the automatic tuning of algorithmic parameters to further enhance the method’s applicability and computational performance.

\section*{Acknowledgments} 

The authors thank the members of Bayes4Helath and CoSInES for fruitful discussions.
DF was supported by the Italian Ministry of University and Research (MUR), 
Department of Excellence project 2023-2027 ReDS 'Rethinking Data Science' - 
Department of Statistics, Computer Science, Applications - University of Florence.
AQW acknowledges financial support from a Discipline Hopping Award awarded by the EPSRC Prob AI Hub.
SG acknowledges support from the European Union (ERC), through the Starting grant ‘PrSc-HDBayLe’ (101076564).
AC was supported by EPSRC grantEP/R018561/1, New Approaches to BayesianData Science: Tackling Challenges from the Health Sciences and by the Royal Society (Dimension Supercharged Projective Sampling project).
GOR acknowledges financial support from UKRI grant EP/Y014650/1
as part of the ERC Synergy project OCEAN,
by EPSRC grants Bayes for Health (R018561), CoSInES (R034710),
PINCODE (EP/X028119/1), and EP/V009478/1.
SR was supported by the Research Council of Norway, Integreat - Norwegian Centre for knowledge-driven machine learning, project number 332645, and by EPSRC grant Bayes4Health,
‘New Approaches to Bayesian Data Science: Tackling Challenges from the Health Sciences’ (EP/R018561/1).
FP was supported by the UKRI grant EP/Y014650/1
as part of the ERC Synergy project OCEAN, and by EPSRC grant Bayes4Health,
‘New Approaches to Bayesian Data Science: Tackling Challenges from the Health Sciences’ (EP/R018561/1). 
PDWK was supported by the UK Medical Research
Council (MC\_UU\_00040/05).

\printbibliography

@article{roberts2015surprising,
  title={Surprising convergence properties of some simple Gibbs samplers under various scans},
  author={Roberts, Gareth O and Rosenthal, Jeffrey S},
  journal={International Journal of Statistics and Probability},
  volume={5},
  number={1},
  pages={51--60},
  year={2015}
}

@article{cancer2012tcga,
  title={Comprehensive molecular portraits of human breast tumours},
  author={{Cancer Genome Atlas Network}},
  journal={Nature},
  volume={490},
  number={7418},
  pages={61--70},
  year={2012},
  doi={10.1038/nature11412},
  publisher={Nature Publishing Group}
}

@article{Lock2013,
    author = {Lock, Eric F. and Dunson, David B.},
    title = {Bayesian consensus clustering},
    journal = {Bioinformatics},
    volume = {29},
    number = {20},
    pages = {2610-2616},
    year = {2013},
    month = {08},
    issn = {1367-4803},
    doi = {10.1093/bioinformatics/btt425},
    url = {https://doi.org/10.1093/bioinformatics/btt425},
    eprint = {https://academic.oup.com/bioinformatics/article-pdf/29/20/2610/48891353/bioinformatics_29_20_2610.pdf},
}

@article{Ascolani2024,
archivePrefix = {arXiv},
arxivId = {2410.00858},
author = {Ascolani, Filippo and Lavenant, Hugo and Zanella, Giacomo},
eprint = {2410.00858},
title = {{Entropy contraction of the Gibbs sampler under log-concavity}},
url = {https://arxiv.org/abs/2410.00858v1},
year = {2024}
}

@article{Goyal2025,
archivePrefix = {arXiv},
arxivId = {2506.22258},
author = {Goyal, Alexander and Deligiannidis, George and Kantas, Nikolas},
eprint = {2506.22258},
title = {{Mixing Time Bounds for the Gibbs Sampler under Isoperimetry}},
url = {http://arxiv.org/abs/2506.22258},
year = {2025}
}

@article{he2016scan,
  title={Scan order in Gibbs sampling: Models in which it matters and bounds on how much},
  author={He, Bryan D and De Sa, Christopher M and Mitliagkas, Ioannis and R{\'e}, Christopher},
  journal={Advances in neural information processing systems},
  volume={29},
  year={2016}
}

@article{andrieu2008tutorial,
  title={A tutorial on adaptive {MCMC}},
  author={Andrieu, Christophe and Thoms, Johannes},
  journal={Statistics and Computing},
  volume={18},
  number={4},
  pages={343--373},
  year={2008},
  publisher={Springer},
  doi={10.1007/s11222-008-9110-y}
}

@article{fiorentini2014efficient,
  title={Efficient {MCMC} sampling in dynamic mixture models},
  author={Fiorentini, Gabriele and Planas, Christophe and Rossi, Alessandro},
  journal={Statistics and Computing},
  volume={24},
  number={1},
  pages={77--89},
  year={2014},
  publisher={Springer},
  doi={10.1007/s11222-012-9354-4}
}

@book{fruhwirth2006finite,
  title={Finite Mixture and {M}arkov Switching Models},
  author={Fr{\"u}hwirth-Schnatter, Sylvia},
  year={2006},
  publisher={Springer},
  address={New York},
  isbn={978-0-387-32909-3},
  doi={10.1007/978-0-387-35768-3}
}

@book{pathria2011statistical,
  title={Statistical Mechanics},
  author={Pathria, R. K. and Beale, Paul D.},
  edition={3rd},
  year={2011},
  publisher={Academic Press},
  address={Boston},
  isbn={978-0-12-382188-1}
}

@article{rousseau2011,
author = {Rousseau, Judith and Mengersen, Kerrie},
title = {Asymptotic behaviour of the posterior distribution in overfitted mixture models},
journal = {Journal of the Royal Statistical Society: Series B (Statistical Methodology)},
volume = {73},
number = {5},
pages = {689-710},
doi = {https://doi.org/10.1111/j.1467-9868.2011.00781.x},
url = {https://rss.onlinelibrary.wiley.com/doi/abs/10.1111/j.1467-9868.2011.00781.x},
eprint = {https://rss.onlinelibrary.wiley.com/doi/pdf/10.1111/j.1467-9868.2011.00781.x},
year = {2011}
}

@article{stephens2000label,
  title={Dealing with label switching in mixture models},
  author={Stephens, Matthew},
  journal={Journal of the Royal Statistical Society: Series B (Statistical Methodology)},
  volume={62},
  number={4},
  pages={795--809},
  year={2000},
  publisher={Wiley},
  doi={10.1111/1467-9868.00265}
}

@article{bardenet2015adaptive,
  title={Adaptive {MCMC} with online relabeling},
  author={Bardenet, R{\'e}mi and Capp{\'e}, Olivier and Fort, Gersende and K{\'e}gl, Bal{\'a}zs},
  journal={Bernoulli},
  volume={21},
  number={3},
  pages={1304--1340},
  year={2015},
  publisher={Bernoulli Society for Mathematical Statistics and Probability},
  doi={10.3150/13-BEJ578}
}

@article{Pearson1894,
author = {Pearson, Karl },
title = {Contributions to the mathematical theory of evolution},
journal = {Philosophical Transactions of the Royal Society of London. (A.)},
volume = {185},
number = {},
pages = {71-110},
year = {1894}
}

@book{Titterington1985,
  title={Statistical Analysis of Finite Mixture Distributions},
  author={Titterington, D.M. and Smith, A.F.M. and Makov, U.E.},
  isbn={9780471907633},
  lccn={85006434},
  series={Applied section},
  url={https://books.google.co.uk/books?id=hZ0QAQAAIAAJ},
  year={1985},
  publisher={Wiley}
}

@article{Richardson1997,
author = {Richardson, Sylvia. and Green, Peter J.},
title = {On Bayesian Analysis of Mixtures with an Unknown Number of Components (with discussion)},
journal = {Journal of the Royal Statistical Society: Series B (Statistical Methodology)},
volume = {59},
number = {4},
pages = {731-792},
year = {1997}
}

@book{Everitt1981,
  title={Finite Mixture Distributions},
  author={Everitt, B. and Hand, D.J.},
  isbn={9780412224201},
  lccn={80041131},
  series={Monographs on Statistics and Applied Probability},
  year={1981},
  publisher={Springer Netherlands},
  doi = {10.1007/978-94-009-5897-5}
}

@book{McLachlan1987,
  title={Mixture Models},
  author={Geoffrey J. McLachlan and Basford, K.E.},
  isbn={9780824776916},
  lccn={87018931},
  series={Statistics:  A Series of Textbooks and Monographs},
  year={1987},
  publisher={Taylor \& Francis}
}

@article{Andrieu2005,
author = {Andrieu, Christophe and Moulines, {\'{E}}ric and Priouret, Pierre},
doi = {10.1137/S0363012902417267},
issn = {03630129},
journal = {SIAM Journal on Control and Optimization},
number = {1},
pages = {283--312},
publisher = {Society for Industrial and Applied Mathematics},
title = {{Stability of Stochastic Approximation under Verifiable Conditions}},
volume = {44},
year = {2005}
}

@article{Hobert1998,
author = {Hobert, James P. and Casella, George},
doi = {10.1080/10618600.1998.10474760},
issn = {15372715},
journal = {Journal of Computational and Graphical Statistics},
number = {1},
pages = {42--60},
title = {{Functional compatibility, markov chains, and gibbs sampling with improper posteriors}},
volume = {7},
year = {1998}
}

@article{Roberts2001,
author = {Roberts, Gareth O. and Rosenthal, Jeffrey S.},
doi = {10.1111/1467-9469.00250},
issn = {03036898},
journal = {Scandinavian Journal of Statistics},
number = {3},
pages = {489--504},
title = {{Markov chains and de-initializing processes}},
volume = {28},
year = {2001}
}

@book{Meyn1993a,
address = {London},
author = {Meyn, S.P. and Tweedie, R.L},
pages = {568},
publisher = {Springer-Verlag},
title = {{Markov Chains and Stochastic Stability}},
year = {1993}
}

@article{latuszynski2013adaptive,
  title={Adaptive Gibbs samplers and related MCMC methods},
  author={{\L}atuszy{\'n}ski, Krzysztof and Roberts, Gareth O and Rosenthal, Jeffrey S},
  journal={The Annals of Applied Probability},
  volume={23},
  number={1},
  pages={66--98},
  year={2013},
  publisher={Institute of Mathematical Statistics}
}

@article{diebolt1994estimation,
  title={Estimation of finite mixture distributions through Bayesian sampling},
  author={Diebolt, Jean and Robert, Christian P},
  journal={Journal of the Royal Statistical Society: Series B (Methodological)},
  volume={56},
  number={2},
  pages={363--375},
  year={1994},
  publisher={Wiley Online Library},
  doi = {10.1111/j.2517-6161.1994.tb01985.x}
}

@article{nobile2004,
 ISSN = {00905364},
 URL = {http://www.jstor.org/stable/3448564},
 author = {Nobile, A.},
 journal = {The Annals of Statistics},
 number = {5},
 pages = {2044--2073},
 publisher = {Institute of Mathematical Statistics},
 title = {On the Posterior Distribution of the Number of Components in a Finite Mixture},
 urldate = {2022-10-24},
 volume = {32},
 year = {2004}
}

@article{miller2018,
author = {Jeffrey W. Miller and Matthew T. Harrison},
title = {Mixture Models With a Prior on the Number of Components},
journal = {Journal of the American Statistical Association},
volume = {113},
number = {521},
pages = {340-356},
year  = {2018},
publisher = {Taylor & Francis},
doi = {10.1080/01621459.2016.1255636},
note ={PMID: 29983475},
URL = {https://doi.org/10.1080/01621459.2016.1255636},
eprint = {https://doi.org/10.1080/01621459.2016.1255636}
}

@article{maceachern1994,
author = { Steven N. MacEachern },
title = {Estimating normal means with a conjugate style dirichlet process prior},
journal = {Communications in Statistics - Simulation and Computation},
volume = {23},
number = {3},
pages = {727-741},
year  = {1994},
publisher = {Taylor & Francis},
doi = {10.1080/03610919408813196},
URL = {https://doi.org/10.1080/03610919408813196},
eprint = {https://doi.org/10.1080/03610919408813196}
}

@article{escobar1995,
 ISSN = {01621459},
 URL = {http://www.jstor.org/stable/2291069},
 author = {Michael D. Escobar and Mike West},
 journal = {Journal of the American Statistical Association},
 number = {430},
 pages = {577--588},
 publisher = {[American Statistical Association, Taylor & Francis, Ltd.]},
 title = {Bayesian Density Estimation and Inference Using Mixtures},
 urldate = {2022-10-24},
 volume = {90},
 year = {1995},
 doi = {10.2307/2291069}
}

@article{neal2000,
 ISSN = {10618600},
 URL = {http://www.jstor.org/stable/1390653},
 author = {Radford M. Neal},
 journal = {Journal of Computational and Graphical Statistics},
 number = {2},
 pages = {249--265},
 publisher = {[American Statistical Association; Taylor & Francis, Ltd.; Institute of Mathematical Statistics; Interface Foundation of America]},
 title = {Markov Chain Sampling Methods for Dirichlet Process Mixture Models},
 urldate = {2022-10-24},
 volume = {9},
 year = {2000}
}

@article{tsallis1988possible,
  title={Possible generalization of Boltzmann-Gibbs statistics},
  author={Tsallis, Constantino},
  journal={Journal of statistical physics},
  volume={52},
  pages={479--487},
  year={1988},
  publisher={Springer}
}

@article{fraley2007bayesian,
  title={Bayesian regularization for normal mixture estimation and model-based clustering},
  author={Fraley, Chris and Raftery, Adrian E},
  journal={Journal of classification},
  volume={24},
  number={2},
  pages={155--181},
  year={2007},
  publisher={Springer}
}

@article{geman1984stochastic,
  title={Stochastic relaxation, Gibbs distributions, and the Bayesian restoration of images},
  author={Geman, Stuart and Geman, Donald},
  journal={IEEE Transactions on pattern analysis and machine intelligence},
  number={6},
  pages={721--741},
  year={1984},
  publisher={IEEE}
}

@article{roberts2009examples,
  title={Examples of adaptive MCMC},
  author={Roberts, Gareth O and Rosenthal, Jeffrey S},
  journal={Journal of computational and graphical statistics},
  volume={18},
  number={2},
  pages={349--367},
  year={2009},
  publisher={Taylor \& Francis}
}

@incollection{geweke1991evaluating,
    author = {Geweke, John},
    isbn = {9780198522669},
    title = {Evaluating the Accuracy of Sampling-Based Approaches to the Calculation of Posterior Moments},
    booktitle = {Bayesian Statistics 4: Proceedings of the Fourth Valencia International Meeting, Dedicated to the memory of Morris H. DeGroot, 1931–1989},
    publisher = {Oxford University Press},
    year = {1992},
    month = {08},
    doi = {10.1093/oso/9780198522669.003.0010},
    url = {https://doi.org/10.1093/oso/9780198522669.003.0010},
    eprint = {https://academic.oup.com/book/0/chapter/422209572/chapter-pdf/52447184/isbn-9780198522669-book-part-10.pdf},
}

@article{zeisel2015cell,
  title={Cell types in the mouse cortex and hippocampus revealed by single-cell RNA-seq},
  author={Zeisel, Amit and Mu{\~n}oz-Manchado, Ana B and Codeluppi, Simone and L{\"o}nnerberg, Peter and La Manno, Gioele and Jur{\'e}us, Anna and Marques, Sueli and Munguba, Hermany and He, Liqun and Betsholtz, Christer and others},
  journal={Science},
  volume={347},
  number={6226},
  pages={1138--1142},
  year={2015},
  publisher={American Association for the Advancement of Science}
}

@inproceedings{prabhakaran2016dirichlet,
  title={Dirichlet process mixture model for correcting technical variation in single-cell gene expression data},
  author={Prabhakaran, Sandhya and Azizi, Elham and Carr, Ambrose and Pe’er, Dana},
  booktitle={International conference on machine learning},
  pages={1070--1079},
  year={2016},
  organization={PMLR}
}

@article{perou2000molecular,
  title={Molecular portraits of human breast tumours},
  author={Perou, Charles M and S{\o}rlie, Therese and Eisen, Michael B and Van De Rijn, Matt and Jeffrey, Stefanie S and Rees, Christian A and Pollack, Jonathan R and Ross, Douglas T and Johnsen, Hilde and Akslen, Lars A and others},
  journal={nature},
  volume={406},
  number={6797},
  pages={747--752},
  year={2000},
  publisher={Nature Publishing Group UK London}
}

@article{sorlie2001gene,
  title={Gene expression patterns of breast carcinomas distinguish tumor subclasses with clinical implications},
  author={S{\o}rlie, Therese and Perou, Charles M and Tibshirani, Robert and Aas, Turid and Geisler, Stephanie and Johnsen, Hilde and Hastie, Trevor and Eisen, Michael B and Van De Rijn, Matt and Jeffrey, Stefanie S and others},
  journal={Proceedings of the National Academy of Sciences},
  volume={98},
  number={19},
  pages={10869--10874},
  year={2001},
  publisher={National Acad Sciences}
}

@article{parker2009supervised,
  title={Supervised risk predictor of breast cancer based on intrinsic subtypes},
  author={Parker, Joel S and Mullins, Michael and Cheang, Maggie CU and Leung, Samuel and Voduc, David and Vickery, Tammi and Davies, Sherri and Fauron, Christiane and He, Xiaping and Hu, Zhiyuan and others},
  journal={Journal of clinical oncology},
  volume={27},
  number={8},
  pages={1160--1167},
  year={2009},
  publisher={American Society of Clinical Oncology}
}

@article{hu2006molecular,
  title={The molecular portraits of breast tumors are conserved across microarray platforms},
  author={Hu, Zhiyuan and Fan, Cheng and Oh, Daniel S and Marron, JS and He, Xiaping and Qaqish, Bahjat F and Livasy, Chad and Carey, Lisa A and Reynolds, Evangeline and Dressler, Lynn and others},
  journal={BMC genomics},
  volume={7},
  pages={1--12},
  year={2006},
  publisher={Springer}
}

@article{nielsen2010comparison,
  title={A comparison of PAM50 intrinsic subtyping with immunohistochemistry and clinical prognostic factors in tamoxifen-treated estrogen receptor--positive breast cancer},
  author={Nielsen, Torsten O and Parker, Joel S and Leung, Samuel and Voduc, David and Ebbert, Mark and Vickery, Tammi and Davies, Sherri R and Snider, Jacqueline and Stijleman, Inge J and Reed, Jerry and others},
  journal={Clinical cancer research},
  volume={16},
  number={21},
  pages={5222--5232},
  year={2010},
  publisher={AACR}
}

@article{cheang2012responsiveness,
  title={Responsiveness of intrinsic subtypes to adjuvant anthracycline substitution in the NCIC. CTG MA. 5 randomized trial},
  author={Cheang, Maggie CU and Voduc, K David and Tu, Dongsheng and Jiang, Shan and Leung, Samuel and Chia, Stephen K and Shepherd, Lois E and Levine, Mark N and Pritchard, Kathleen I and Davies, Sherri and others},
  journal={Clinical Cancer Research},
  volume={18},
  number={8},
  pages={2402--2412},
  year={2012},
  publisher={AACR},
  doi = {10.1158/1078-0432.CCR-11-2956}
}

@article{sorlie2003repeated,
  title={Repeated observation of breast tumor subtypes in independent gene expression data sets},
  author={S{\o}rlie, Therese and Tibshirani, Robert and Parker, Joel and Hastie, Trevor and Marron, James Stephen and Nobel, Andrew and Deng, Shibing and Johnsen, Hilde and Pesich, Robert and Geisler, Stephanie and others},
  journal={Proceedings of the national academy of sciences},
  volume={100},
  number={14},
  pages={8418--8423},
  year={2003},
  publisher={National Acad Sciences}
}

@article{broet2002bayesian,
  title={Bayesian hierarchical model for identifying changes in gene expression from microarray experiments},
  author={Bro{\"e}t, Philippe and Richardson, Sylvia and Radvanyi, Fran{\c{c}}ois},
  journal={Journal of Computational Biology},
  volume={9},
  number={4},
  pages={671--683},
  year={2002},
  publisher={Mary Ann Liebert, Inc.},
  doi = {10.1089/106652702760277381}
}

@article{fraley2002model,
  title={Model-based clustering, discriminant analysis, and density estimation},
  author={Fraley, Chris and Raftery, Adrian E},
  journal={Journal of the American statistical Association},
  volume={97},
  number={458},
  pages={611--631},
  year={2002},
  publisher={Taylor \& Francis}
}

@article{schlattmann1993mixture,
  title={Mixture models and disease mapping},
  author={Schlattmann, Peter and B{\"o}hning, Dankmar},
  journal={Statistics in medicine},
  volume={12},
  number={19-20},
  pages={1943--1950},
  year={1993},
  publisher={Wiley Online Library}
}

@article{green2002hidden,
  title={Hidden Markov models and disease mapping},
  author={Green, Peter J and Richardson, Sylvia},
  journal={Journal of the American statistical association},
  volume={97},
  number={460},
  pages={1055--1070},
  year={2002},
  publisher={Taylor \& Francis}
}

@article{jedidi1997finite,
  title={Finite-mixture structural equation models for response-based segmentation and unobserved heterogeneity},
  author={Jedidi, Kamel and Jagpal, Harsharanjeet S and DeSarbo, Wayne S},
  journal={Marketing Science},
  volume={16},
  number={1},
  pages={39--59},
  year={1997},
  publisher={INFORMS}
}

@article{allenby1999dynamic,
  title={A dynamic model of purchase timing with application to direct marketing},
  author={Allenby, Greg M and Leone, Robert P and Jen, Lichung},
  journal={Journal of the American Statistical Association},
  volume={94},
  number={446},
  pages={365--374},
  year={1999},
  publisher={Taylor \& Francis},
  doi = {10.1080/01621459.1999.10474127}
}

@article{hamilton1989new,
  title={A new approach to the economic analysis of nonstationary time series and the business cycle},
  author={Hamilton, James D},
  journal={Econometrica: Journal of the econometric society},
  pages={357--384},
  year={1989},
  publisher={JSTOR}
}

@article{fruhwirth2001markov,
  title={Markov chain Monte Carlo estimation of classical and dynamic switching and mixture models},
  author={Fr{\"u}hwirth-Schnatter, Sylvia},
  journal={Journal of the American Statistical Association},
  volume={96},
  number={453},
  pages={194--209},
  year={2001},
  publisher={Taylor \& Francis}
}

@article{weigend2000predicting,
  title={Predicting daily probability distributions of S\&P500 returns},
  author={Weigend, Andreas S and Shi, Shanming},
  journal={Journal of Forecasting},
  volume={19},
  number={4},
  pages={375--392},
  year={2000},
  publisher={Wiley Online Library}
}

@article{kaufmann2002bayesian,
  title={Bayesian analysis of switching ARCH models},
  author={Kaufmann, Sylvia and Fr{\"u}hwirth-Schnatter, Sylvia},
  journal={Journal of Time Series Analysis},
  volume={23},
  number={4},
  pages={425--458},
  year={2002},
  publisher={Wiley Online Library}
}

@article{gelfand1990sampling,
  title={Sampling-based approaches to calculating marginal densities},
  author={Gelfand, Alan E and Smith, Adrian FM},
  journal={Journal of the American statistical association},
  volume={85},
  number={410},
  pages={398--409},
  year={1990},
  publisher={Taylor \& Francis}
}

@article{amit1991comparing,
  title={Comparing sweep strategies for stochastic relaxation},
  author={Amit, Yali and Grenander, Ulf},
  journal={Journal of multivariate analysis},
  volume={37},
  number={2},
  pages={197--222},
  year={1991},
  publisher={Elsevier},
  doi = {10.1016/0047-259X(91)90080-L}
}

@article{fishman1996coordinate,
  title={Coordinate selection rules for Gibbs sampling},
  author={Fishman, George S},
  journal={The Annals of Applied Probability},
  volume={6},
  number={2},
  pages={444--465},
  year={1996},
  publisher={Institute of Mathematical Statistics},
  doi = {10.1214/aoap/1034968139}
}

@article{roberts1997updating,
  title={Updating schemes, correlation structure, blocking and parameterization for the Gibbs sampler},
  author={Roberts, Gareth O and Sahu, Sujit K},
  journal={Journal of the Royal Statistical Society Series B: Statistical Methodology},
  volume={59},
  number={2},
  pages={291--317},
  year={1997},
  publisher={Oxford University Press}
}

@article{haario2001adaptive,
  title={An adaptive Metropolis algorithm},
  author={Haario, Heikki and Saksman, Eero and Tamminen, Johanna},
  year={2001}
}

@article{roberts2007coupling,
  title={Coupling and ergodicity of adaptive Markov chain Monte Carlo algorithms},
  author={Roberts, Gareth O and Rosenthal, Jeffrey S},
  journal={Journal of applied probability},
  volume={44},
  number={2},
  pages={458--475},
  year={2007},
  publisher={Cambridge University Press}
}

@article{carpenter2017stan,
  title={Stan: A probabilistic programming language},
  author={Carpenter, Bob and Gelman, Andrew and Hoffman, Matthew D and Lee, Daniel and Goodrich, Ben and Betancourt, Michael and Brubaker, Marcus and Guo, Jiqiang and Li, Peter and Riddell, Allen},
  journal={Journal of statistical software},
  volume={76},
  pages={1--32},
  year={2017}
}

@article{chaumeny2022bayesian,
  title={Bayesian nonparametric mixture inconsistency for the number of components: How worried should we be in practice?},
  author={Chaumeny, Yannis and Moris, Johan van der Molen and Davison, Anthony C and Kirk, Paul DW},
  journal={arXiv preprint arXiv:2207.14717},
  year={2022}
}

@article{dowsett2013comparison,
  title={Comparison of PAM50 risk of recurrence score with onco type DX and IHC4 for predicting risk of distant recurrence after endocrine therapy},
  author={Dowsett, Mitch and Sestak, Ivana and Lopez-Knowles, Elena and Sidhu, Kalvinder and Dunbier, Anita K and Cowens, J Wayne and Ferree, Sean and Storhoff, James and Schaper, Carl and Cuzick, Jack},
  journal={Journal of Clinical Oncology},
  volume={31},
  number={22},
  pages={2783--2790},
  year={2013},
  publisher={American Society of Clinical Oncology}
}

@article{gnant2014predicting,
  title={Predicting distant recurrence in receptor-positive breast cancer patients with limited clinicopathological risk: using the PAM50 Risk of Recurrence score in 1478 postmenopausal patients of the ABCSG-8 trial treated with adjuvant endocrine therapy alone},
  author={Gnant, M and Filipits, M and Greil, Richard and Stoeger, H and Rudas, M and Bago-Horvath, Z and Mlineritsch, Brigitte and Kwasny, W and Knauer, M and Singer, C and others},
  journal={Annals of oncology},
  volume={25},
  number={2},
  pages={339--345},
  year={2014},
  publisher={Elsevier}
}

@article{prat2015clinical,
  title={Clinical implications of the intrinsic molecular subtypes of breast cancer},
  author={Prat, Aleix and Pineda, Estela and Adamo, Barbara and Galv{\'a}n, Patricia and Fern{\'a}ndez, Aranzazu and Gaba, Lydia and D{\'\i}ez, Marc and Viladot, Margarita and Arance, Ana and Mu{\~n}oz, Montserrat},
  journal={The Breast},
  volume={24},
  pages={S26--S35},
  year={2015},
  publisher={Elsevier}
}

@article{nielsen2014analytical,
  title={Analytical validation of the PAM50-based Prosigna Breast Cancer Prognostic Gene Signature Assay and nCounter Analysis System using formalin-fixed paraffin-embedded breast tumor specimens},
  author={Nielsen, Torsten and Wallden, Brett and Schaper, Carl and Ferree, Sean and Liu, Shuzhen and Gao, Dongxia and Barry, Garrett and Dowidar, Naeem and Maysuria, Malini and Storhoff, James},
  journal={BMC cancer},
  volume={14},
  number={1},
  pages={177},
  year={2014},
  publisher={Springer}
}

@article{hubert1985comparing,
  title={Comparing partitions},
  author={Hubert, Lawrence and Arabie, Phipps},
  journal={Journal of classification},
  volume={2},
  number={1},
  pages={193--218},
  year={1985},
  publisher={Springer}
}

\newpage

\appendix

\section{Alternative discomfort functions} \label{alternatives}

In this section, we discuss various options for the discomfort functions implemented in the model. As highlighted in the manuscript, these functions need to be continuous and exhibit a decreasing trend: lower probabilities of assignment to a given component correspond to higher discomfort values, and vice versa.

\subsection{Generalized entropy}

\citet{tsallis1988possible} introduced a generalised form of the entropy function, defined as follows:

\begin{equation}
    D_{it}^q = \frac{1-\sum_{k = 1}^K p_{i,k,t}^q}{q-1} \, ,
    \label{eq:genEnt}
\end{equation}

\noindent where $q > 0$ and the probability $p_{i,k,t}$ represents the posterior probability that observation $i$ belongs to component $k$ at iteration $t$. In Equation \eqref{eq:genEnt}, for $q = 0$ we obtain the classical entropy function \citep{tsallis1988possible}. Instead of calculating discomfort based on the sum of probabilities across all component, we can focus solely on the component assigned to the observation at time $t$. In this approach, the discomfort of observation $i$ depends exclusively on the probability of $i$ belonging to the assigned component $z_i$ at time $t$:
\begin{equation}
    D_{it}^q = \frac{1- p_{i,z_i,t}^q}{q-1} \, , 
    \label{eq:parEnt}
\end{equation}
\noindent where $p_{i, z_i, t}$ represents the posterior probability that observation $i$ is allocated to its current component $z_i$ at iteration $t$. In Equation \eqref{eq:parEnt}, for $q = 0$ we obtain the ``partial'' entropy function.

\subsection{Pareto} \label{Pareto}
The Pareto distribution models phenomena where a small proportion of causes lead to a large proportion of effects, such as income distribution or extreme events. Defined for $p_{i,z_i,t} \geq p_m$ with scale $p_m$ and shape $\alpha > 0$, its probability density function is:
\begin{align*}
    D^{(p_m, \alpha)}_{it} =
    \begin{cases} 
    \alpha \frac{p_m^\alpha}{p_{i,z_i,t}^{\alpha + 1}} & \text{for } p_{i,z_i,t} \geq p_m, \\
    0 & \text{for } p_{i,z_i,t} < p_m,
    \end{cases}
\end{align*}
characterizing heavy tails that suit modeling rare but impactful occurrences.

\subsection{Weibull} \label{Weibull}

The Weibull distribution is widely used to model lifetimes, reliability, and failure rates. Defined by a shape parameter $k > 0$ and a scale parameter $\lambda > 0$, its probability density function is:

\begin{equation*}
    D_{it}^{(\lambda, k)} = \frac{k}{\lambda} \left( \frac{p_{i,z_i,t}}{\lambda} \right)^{k-1} e^{-\left( \frac{p_{i,z_i,t}}{\lambda} \right)^k},
\end{equation*}
This distribution allows to model various failure behaviors: $k < 1$ indicates decreasing failure rates, $k = 1$ corresponds to the exponential distribution, and $k > 1$ implies increasing failure rates. In our case, we can use either $k = 1$ or $k < 1$.

\subsection{Hyperbolic scaling}

The function: 

\begin{equation*}
    D_{it}^{\lambda} = \left(\frac{1}{p_{i,z_i,t}}\right)^{\lambda} \,  ,
\end{equation*}

\noindent assigning weights or scaling values based on the inverse of a probability  $p_{i, z_i, t}$---higher $\lambda > 0$ values amplify the effect of small probabilities, while lower values reduce this emphasis.

\section{Practical tuning of \texorpdfstring{$\Lambda$}{Lambda}}
\label{app:tuning}

To clarify how the tuning parameter $\Lambda$ can be selected in practice, we outline the procedure using a simple illustrative example. Specifically, consider a mixture of three clusters with means $\mu_1 = (18.3, \, 20.0)$, $\mu_2 = (16.7, \, 13.3)$, and $\mu_3 = (20.0, \, 13.3)$, and diagonal precision matrices with diagonals $\mbox{diag}(\Sigma_1) = (1/2,\, 1/2)$, $\mbox{diag}(\Sigma_2) = (1, \, 1)$, and $\mbox{diag}(\Sigma_3) = (1/2, \, 1/2)$. 
We generate a sample of $n = 1000$ observations from this mixture and run DIG over 20 independent replicas.

Figure~\ref{fig:lambdaBoundary} displays two plots showing the behaviour of the $\lambda$ parameter under two different upper bounds: $\Lambda = 50$ and $\Lambda = 100$, based on the previously described synthetic dataset with $K = 3$. In the case where $\Lambda = 50$ (Figure~\ref{fig:lambdaBoundary_50}), we observe that before the intersection point $s$, the parameter $\lambda$ consistently reaches the maximum allowed value, remaining fixed at $\lambda = 50$. By design, after the intersection point $s$, $\lambda$ is set to 1. 

On the other hand, in the (RHS) plot corresponding to a higher bound $\Lambda = 100$, the $\lambda$ parameter is able to take on different values from the upper bound. 
Note that the algorithm gives a warning if $\lambda$ touches the upper bound $\Lambda$ for more than 25\% of the iterations of the adaptation phase.

\begin{figure}[htbp]
    \centering
    \begin{subfigure}[b]{0.48\textwidth}
        \includegraphics[width=\textwidth]{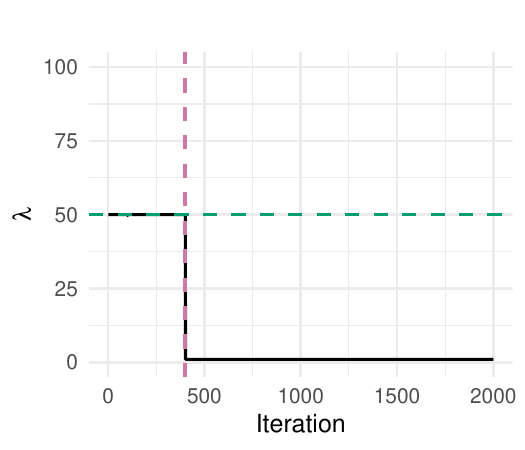}
        \caption{Upper bound $\Lambda = 50$.}
        \label{fig:lambdaBoundary_50}
    \end{subfigure}
    \hfill
    \begin{subfigure}[b]{0.48\textwidth}
        \includegraphics[width=\textwidth]{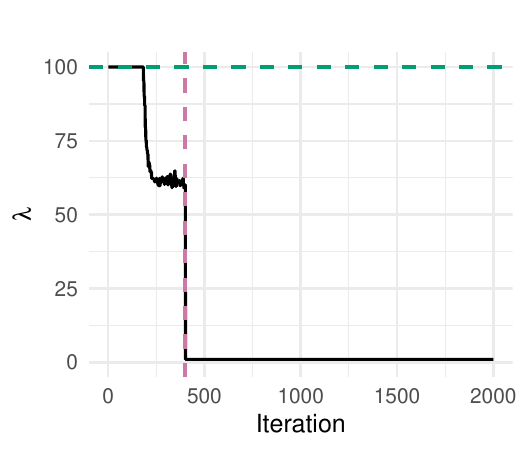}
        \caption{Upper bound $\Lambda = 100$.}
        \label{fig:lambdaBoundary_100}
    \end{subfigure}
    \caption{Trace plots of the $\lambda$ parameter (black line) over the MCMC iterations, based on data simulated from a well-specified 3-component Gaussian mixture with $n = 1000$, $d = 2$, and $K = 3$ (Appendix \ref{app:tuning}). The green dashed horizontal line marks the fixed upper bound $\Lambda = 50, 100$ (LHS, RHS). The magenta dashed vertical line indicates the iteration $s$ at which $\lambda$ is set to 1.}
    \label{fig:lambdaBoundary}
\end{figure}

\section{Free parameters and sensitivity analysis} \label{app:free}

Our algorithm contains several free parameters that need to be specified. These parameters include: $m$, which determines the number of observations to be updated at each MCMC iteration; $\xi$, which specifies how often the allocation probability matrix is updated; the decay parameter $\lambda$; $s$, representing the intersection point of the weight functions; and $a$, which controls how “close” the hyperbolic weight functions are.

The parameter $m \in \{1,2,\ldots,n \}$, can be adapted to the problem. It could simply be taken as a multiple of the number of cores and/or GPUs available. Or else it could be chosen by counting the number of individuals whose discomfort score is above a user-defined threshold, such that on average we update exactly the number of individuals that are likely to switch to a different cluster. In this work, we choose to set $m$ as a proportion of the total number of observations $n$. As discussed in Section \ref{Hyper}, selecting $m$ too small relative to $n$ can lead to slower convergence compared to the SSG method. On the other hand, when dealing with large-scale problems, comparing with SSG becomes more meaningful, since SSG updates all allocations at each iteration. Conversely, choosing $m$ too large with respect to $n$ would eliminate the computational advantages of our method, bringing its cost close to that of SSG. 

We conducted a sensitivity analysis on the schedule used to update both the allocation probability matrix ($\xi$) and the frequency at which the parameter $\lambda$ is updated, since the latter is adjusted every time the former is updated. Specifically, we explored five different schedules: $1-5-10$, $3-6-10$, $5-8-10$, $7-10-15$, and $10-15-20$. The results for each of these configurations are reported in Table \ref{sensXi}. For all simulations, we used the dataset introduced by \citet{miller2018}, varying the sample size n and the dimension d, while keeping the remaining parameters consistent with those presented in Section \ref{simStud}.

\begin{table}[htbp]
\centering
\renewcommand{\arraystretch}{1.4}
\begin{tabular}{ccccccc}
\hline
         & \multicolumn{2}{c}{$n = 1000, \, d = 2$} & \multicolumn{2}{c}{$n = 3000,\,  d = 5$} & \multicolumn{2}{c}{$n = 5000, \, d = 10$} \\ \hline
            & T2C            & CLL             & T2C               & CLL            & T2C                & CLL             \\
$1$-$5$-$10$    & 0.03 (0.04)   & -2587 (14)     &  0.06 (0.03)     & -18653 (13)    &  0.26 (0.16)      &  -64671 (16)      \\
$3$-$6$-$10$    & \textbf{0.01} (0.01)   & -2597 (2)     &  \textbf{0.05} (0.01)     & -18652 (14)    &  \textbf{0.12} (0.01)      &  -64672 (12)       \\
$5$-$8$-$10$    & 0.02 (0.03)   & -2590 (11)     &  \textbf{0.05} (0.01)     & -18652 (10)    &  0.13 (0.03)      &  -64668 (20)                 \\
$7$-$10$-$15$   & \textbf{0.01} (0.01)   & -2590 (14)     &  0.85 (2.52)     & -18673 (45)    &  1.32 (3.78)      & -64679 (20)                   \\
$10$-$15$-$20$  &  0.02 (0.02)  & -2587 (13)     &  0.84 (2.39)     & -18685 (90)    &  0.15 (0.03)      &  -64663 (14)      \\ \hline
\end{tabular}
\caption{Miller-Harrison dataset sensitivity analysis under different values of $n$ and $d$: Mean Complete Log-Likelihood (CLL) and Mean Time to Convergence (T2C) over 20 replicas (with standard errors reported in parentheses) across different schedules. T2C values have been approximated to two decimal places, and CLL values have been rounded to the nearest integer to improve readability. The best performance in terms of T2C is highlighted in bold.}
\label{sensXi}
\end{table}

In Table \ref{sensXi} we observe that the impact on the convergence speed is primarily determined by the first value in the schedule, as it is assumed that after the initial iterations, the observations are assigned to the correct components. Table \ref{sensXi} shows that the algorithm remains fairly robust to changes in the schedule.

The inverse temperature $\lambda$ can instead be updated according to the equation below:

\begin{equation}
\label{eq:lambda_schedule2}
\lambda_t = 
\begin{cases} 
\left(\frac{1}{L + 1}\sum_{l=1}^L \lambda_{t-l} + \hat{\lambda}_t\right) \vee 1, & t \leq s \\
1, & t > s
\end{cases}
\end{equation}
where $\vee$ denotes the binary maximum operation, and $\hat{\lambda}_t$ solves $\text{ESS}(\lambda_t) = c \cdot m$. Here $c$ is another potential regularization parameter. The higher we set $c$, the higher the effective number of particles / components of $\bm \alpha$ will be. In the equation above, we update $\lambda_t$ for $t<s$ according to a moving average that includes components up to lag $L$ in the past. With $L>0$, Equation \eqref{eq:lambda_schedule2} yields a $\lambda_t$ that is smoother than it would normally be. We verified empirically that the algorithm does not seem sensitive to the value of $L$ and $c$, providing that they are selected to be small. However, by Occam's Razor, we do not use Equation \eqref{eq:lambda_schedule2}. In our work we use the much simpler Equation \eqref{eq:lambda_schedule}.

The parameters $s$ is selected as the $99 \%$ quantile of the Normal approximation to a Negative Binomial distribution, as described in Section \ref{sec:s}. Other choices include values from the symmetrical confidence region
\begin{equation} \label{interval}
    \left[\mathbf{1}\left(\mu - \sigma Z_{1-\alpha/2} \geq 0\right) \cdot \left(\mu - \sigma Z_{1-\alpha/2}\right) \hspace{0.1cm} ,  \hspace{0.2cm} \mu + \sigma Z_{1-\alpha/2} \right]
\end{equation}
where $\mathbf{1}$ represents the indicator function, with an arbitrary level of confidence, not necessarily $.01$.

Indeed, the user can add---or remove---an arbitrary number of standard deviations from the mean, in order to find a value that works best for their dataset. However, we do recommend to not start collapsing the weights until the algorithm has reached equilibrium, otherwise it may take a long time for the algorithm to converge to the right clustering stricture. Therefore we urge caution especially when decreasing the value of $s$.

\section{Weight functions} \label{sec:weFunction}

This section analyzes the choice of mixed weights versus using a single functional form for all iterations. An initial option involves adopting the following weight functions:
\begin{equation} \label{polfun1}
f(t) = \frac{t}{t+1} \quad , \quad g(t) = \frac{1}{t+1} \, .
\end{equation}
However, this formulation allows the function $g(t)$ to influence the Gibbs sampler for only one iteration, as shown in the left panel of Figure \ref{fig:hyperpoly}.\\
To mitigate this issue, a modification of the intersection point between the two curves is considered, introducing a parameter $s$ that regulates its position. The weights are then redefined as:
\begin{equation} \label{polfuns}
f(t) = \frac{t}{t+s} \quad , \quad g(t) = \frac{s}{t+s} \, .
\end{equation}
Adopting a sufficiently large value for $s$ allows the function $g(t)$ to guide the MCMC for a greater number of iterations, as shown in the left panel of Figure \ref{fig:hyperpoly}. However, this choice entails a potential delay in convergence, as the functions $f(t)$ and $g(t)$ tend to approach the asymptotes slowly, risking an excessive prolongation of the process.

Another alternative is the use of weight functions based on the hyperbolic tangent, defined as
\begin{equation*}
    f(t) = \frac{1}{2} \left(1+\tanh{\left(\frac{t-s}{a} \right)}\right) \hspace{0.1cm}, \hspace{0.3cm} g(t) = \frac{1}{2} \left(1-\tanh{\left(\frac{t-s}{a} \right)}\right) \, ,
\end{equation*}
and shown in the right panel of Figure \ref{fig:hyperpoly}.

\begin{figure}[htbp]
    \centering
    \begin{subfigure}[b]{0.48\textwidth}
        \includegraphics[width=\textwidth]{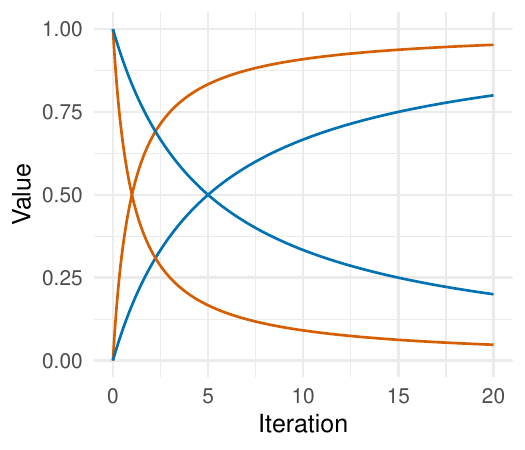}
        \caption{Visualization of the polynomial weight functions from Equations \eqref{polfun1} and \eqref{polfuns} for $s = 1$ (blue curves) and $s = 5$ (orange curves), showing how the parameter $s$ affects their dynamics across iterations.}
    \end{subfigure}
    \hfill
    \begin{subfigure}[b]{0.48\textwidth}
        \includegraphics[width=\textwidth]{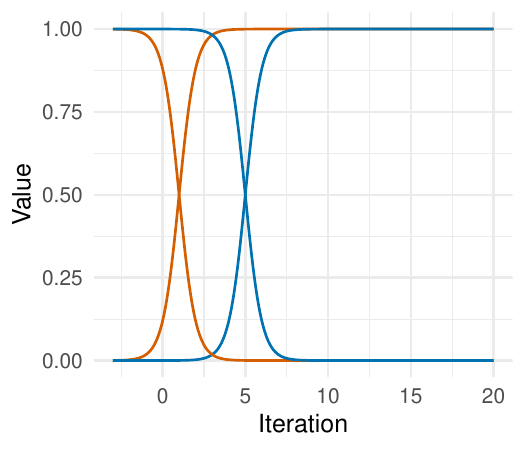}
        \caption{Visualization of the functions from Equation \eqref{eq:wFunHyper} for $s = 1$ (blue curves) and $s = 5$ (orange curves), with $a = 1$, showing how the parameter $s$ influences their behavior over iterations.}
    \end{subfigure}
    \caption{Behavior of polynomial (left) and hyperbolic (right) weight functions from Equations \eqref{polfun1} and \eqref{polfuns}, plotted for $s \in {1, 5}$ with $a = 1$, showing the role of $s$ in the transition dynamics.}
    \label{fig:hyperpoly}
\end{figure}

The issue with the hyperbolic weights defined in Equation \eqref{eq:wFunHyper} is twofold: on one hand, they are incompatible with the stochastic approximation framework; on the other, their rapid saturation towards asymptotes can cause the algorithm to become prematurely stuck in suboptimal configurations.

At the same time, relying solely on polynomial weights would result in an initially excessively slow dynamic, reducing the efficiency of exploration and delaying convergence.

To overcome these limitations, common to both weighting schemes, we adopted a hybrid strategy: for $t \leq s$, the algorithm is guided by hyperbolic weights, while for $t > s$, polynomial weights are employed, as detailed in the manuscript.

In defining the polynomial component, we deliberately avoided using the functions in Equation \eqref{polfuns}, since an intersection point $s$ set too far ahead would cause a similar issue as with the hyperbolic weights---a rapid stiffening of the dynamic that hinders effective exploration of the configuration space.

Therefore, we calibrated the initial derivative of the polynomial functions at the intersection point $s$ to match that of the functions in Equation \eqref{polfun1}. This choice ensures a more gradual and controlled decay, allowing the algorithm sufficient time to modify assignments and progress toward convergence, preventing premature stabilization of the dynamics.

\newpage
\section{Complete DIG algorithm}\label{completeAlgorithm}

\begin{algorithm}[ht!]
\SetKwInOut{Input}{input}\SetKwInOut{Output}{output}
\SetAlgoLined
\Input{Initial sampling weights $ \boldsymbol{\alpha}_0= (\alpha_{0,1}, \ldots, \alpha_{0,n}) $; total iterations $T$; subset size $m$; upper bound $\Lambda$; transition parameter $a$.}

\BlankLine

Compute transition point $s = \mu + \sigma Z_{0.99}$ according to Equation \eqref{eq:s}\;
Initialise allocation probability matrix\;

\For{$t = 1, \ldots, T$}{
\If{allocation matrix update is due at iteration $t$}{
Update allocation probability matrix\;
}

\eIf {$t \leq s$} {
Set $\lambda_t = \hat \lambda_t$\;
Compute hyperbolic tangent weights $f(t), g(t)$ using Equation \eqref{eq:wFunHyper}\;
}{
Set $\lambda_t = 1$\;
Compute polynomial weights $f(t), g(t)$ using Equation \eqref{eq:wFunPol}\;
}

Update $\boldsymbol{\alpha}_t = f(t) \cdot \boldsymbol{\alpha}_{t-1} + g(t) \cdot  \boldsymbol{D}_t^{\lambda_t}$\;
Normalise: $\alpha_{i,t}' = \alpha_{i,t} / \sum_{j=1}^n \alpha_{j,t}$ for all $i$\;

Sample $m$ distinct indices $\{i_1, \ldots, i_m\}$ without replacement using probabilities $\boldsymbol{\alpha}_t'$\;

\For{$j = 1, \ldots, m$}{
$z_{i_j,t} \sim p(z_{i_j}\mid\boldsymbol{\Theta}_{t-1} , \boldsymbol{\pi}_{t-1}, \boldsymbol{X})$\;
}
$\boldsymbol{\pi}_{t} \sim p(\boldsymbol{\pi}\mid\boldsymbol{\Theta}_{t-1} ,\boldsymbol{z}_{t} , \boldsymbol{X})$\;
$\boldsymbol{\Theta}_{t} \sim p(\boldsymbol{\Theta}\mid\boldsymbol{\pi}_{t} ,\boldsymbol{z}_{t} , \boldsymbol{X})$\;
} 

\caption{DIG: Discomfort-informed Gibbs sampler for finite mixture models}
\label{alg:ARSG4}
\end{algorithm}

Note that the allocation matrix update schedule $\xi$ follows the adaptive schedule described in Section \ref{allUp} (every 3 iterations for the first 25\% of the run, every 6 iterations for 25-50\%, and every 10 iterations for the final 50\%).

\section{Complete log-likelihood from Section \ref{simStud} }

\begin{table}[htbp]
    \centering
\renewcommand{\arraystretch}{1.2}
\begin{tabular}{@{}clcccc@{}}
\toprule
& & \multicolumn{4}{c}{\textbf{Mean Complete Log-Likelihood}} \\
\cmidrule(lr){3-6}
\textbf{Sample Size} & \textbf{Method} & $d = 2$ & $d = 5$ & $d = 10$ & $d = 20$ \\
\midrule
\multirow{3}{*}{$n = 1000$} 
& SSG & -2596 (4.915) & -6213 (4.015) & -12873 (3.900) & -26844 (4.109) \\
& RSG & \textbf{-2591} (3.781) & \textbf{-6212} (3.825) & \textbf{-12871} (3.185) & -26844 (3.123) \\
& DIG & -2597 (2.254) & -6213 (3.791) & -12874 (3.311) & \textbf{-26842} (3.523) \\
\midrule
\multirow{3}{*}{$n = 5000$}
& SSG & -12755 (9.748) & -31034 (8.814) & -64815 (8.679) & -134371 (8.448) \\
& RSG & -12754 (9.293) & \textbf{-31029} (7.010) & -645817 (7.402) & -134371 (9.439) \\
& DIG & \textbf{-12751} (11.368) & -31032 (8.493) & \textbf{-64814} (7.537) & \textbf{-134365} (7.880) \\
\midrule
\multirow{3}{*}{$n = 10000$}
& SSG & -25343 (13.977) & -62037 (12.616) & -129190 (11.902) & -268514 (11.893) \\
& RSG & \textbf{-25334} (13.689) & \textbf{-62036} (11.198) & \textbf{-129186} (9.154) & -268512 (9.323) \\
& DIG & -25336 (13.562) & -62038 (13.592) & -129192 (9.115) & \textbf{-268510} (11.848) \\
\bottomrule
\end{tabular}
    \caption{Miller-Harrison synthetic datasets across different scenarios: Mean Complete Log-Likelihood (CLL) computed over 20 replications, with standard errors reported in parentheses, for different scenarios and algorithms. Mean CLL values have been rounded to the nearest thousandth to improve readability. The best performance in terms of CLL is highlighted in bold.}
    \label{resultsARI}
\end{table}

The CLL for the experiments on Gaussian mixtures from Section \ref{simStud} is shown below. The CLL values are averaged over 20 replicas for the SSG, RSG, and DIG algorithms, across different values of sample size $n$ and dimensionality $d$. All the values are very similar, suggesting that all the algorithms converge to the same configuration.

\newpage

\section{Heat maps from Section \ref{sec:real}} \label{app:heatmaps}

\subsection{Section \ref{sec:real-soma}}

Figure \ref{fig:matrixGENE} below shows the PSM and heat map from the analysis of gene expression data from single cell RNA sequencing of somatosensory cells based on \citet{zeisel2015cell}, from Section \ref{sec:real-soma}. 
This dataset has $d=19972$ genes in $n=3005$ cells, but following \citet{chaumeny2022bayesian}, we limit the analysis to the $d = 10$ genes with the highest standard deviations. 
We model the data assuming $K^{\ast} = K=7$, and a spherical covariance structure.

\begin{figure}[!htb]
    \centering
    \begin{subfigure}[t]{\textwidth}
        \includegraphics[scale = 0.42]{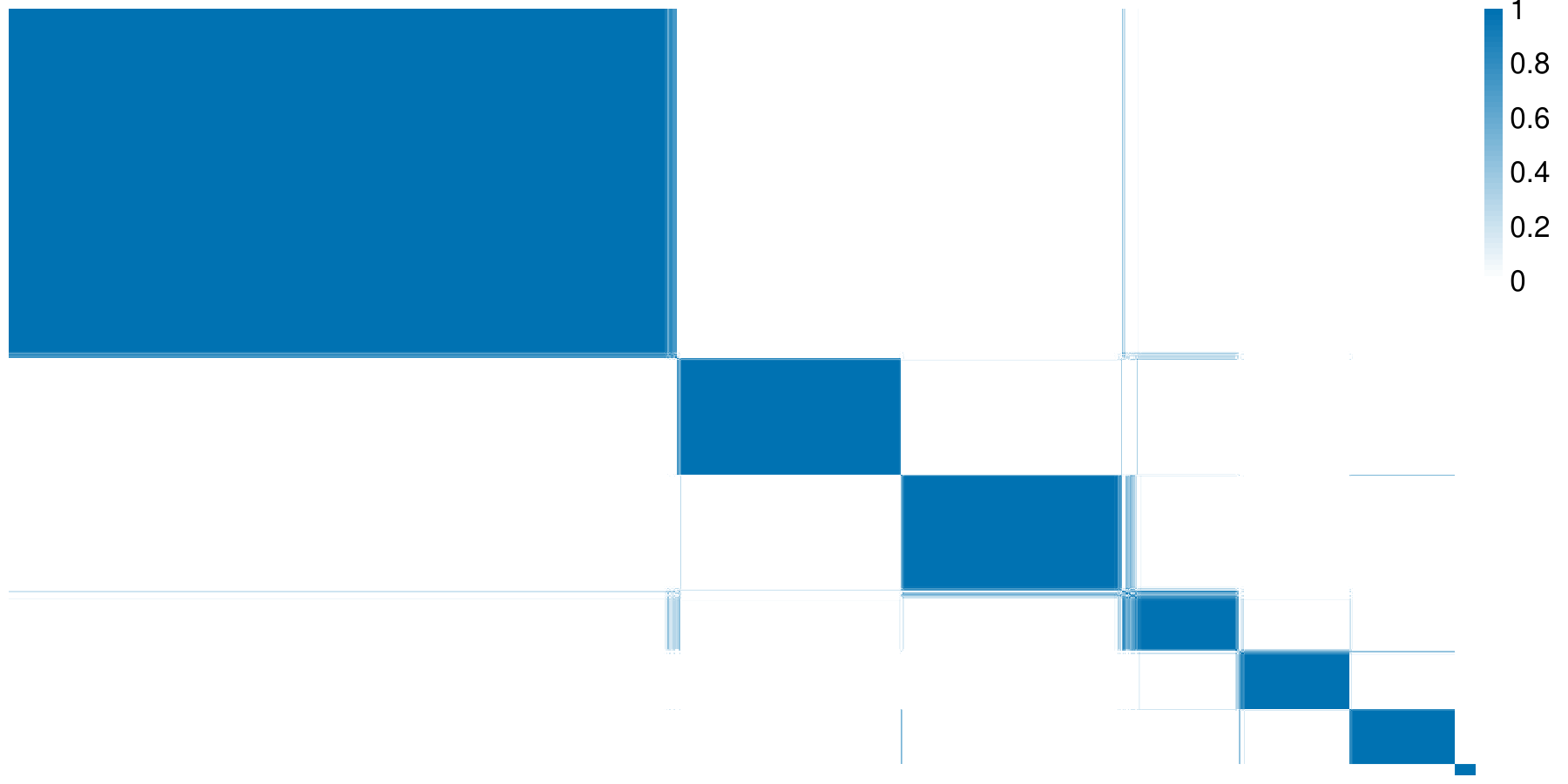} 
        \caption{Posterior similarity matrix.}
        \label{fig:psm_gene}
    \end{subfigure}
    
    \vspace{5em}  
    
    \begin{subfigure}[t]{\textwidth}
        \includegraphics[scale = 0.48]{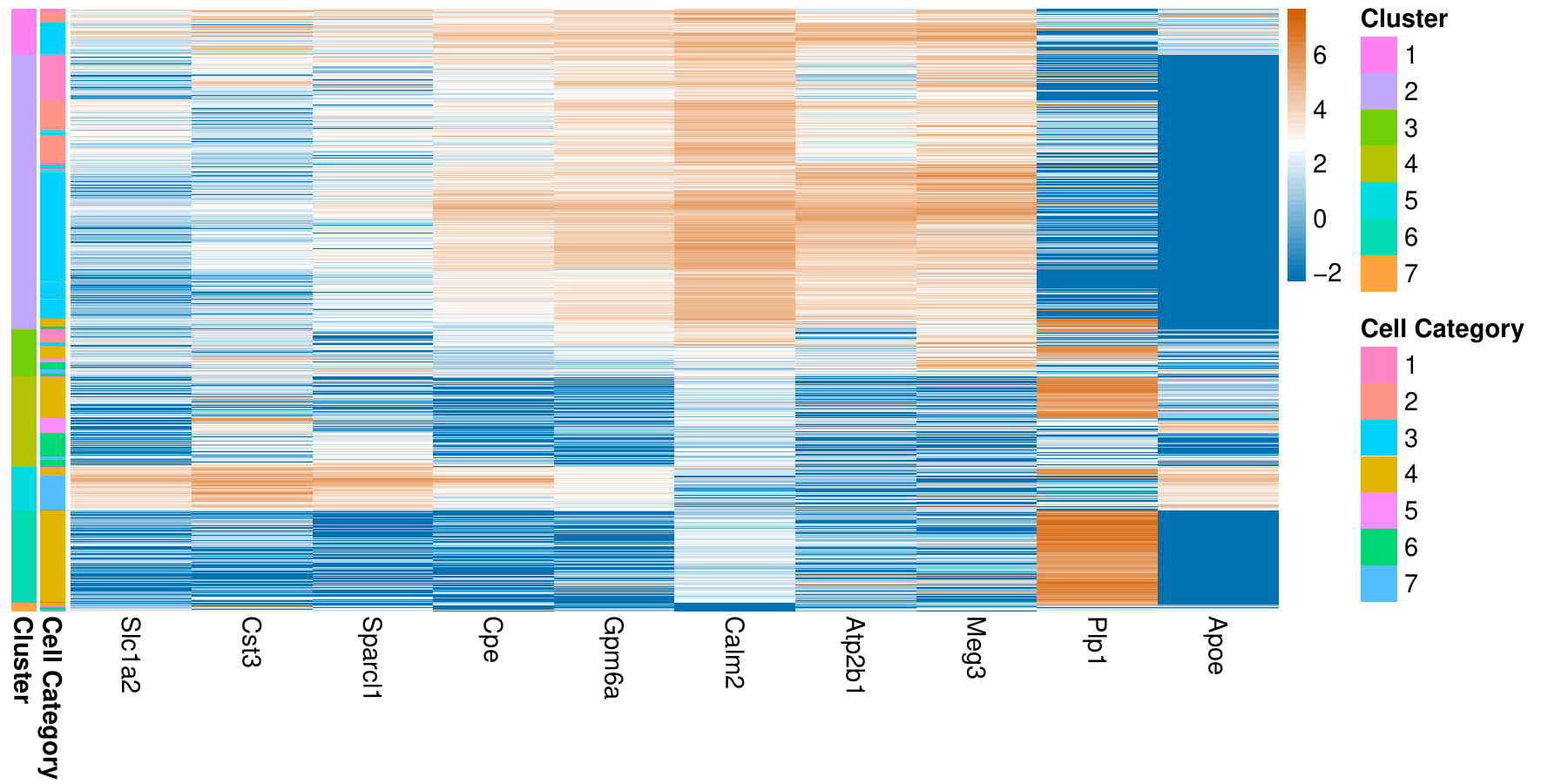}
        \caption{Gene expression heatmap across seven major brain cell categories: somatosensory pyramidal neurons (1), CA1~pyramidal~neurons (2), interneurons (3), oligodendrocytes (4), astrocytes (5), microglia (6), and endothelial cells (7).}
        \label{fig:pheat_gene}
    \end{subfigure}

    \vspace{4em}
    
    \caption{Mouse Cortex gene expression clustering results (n = 3005 cells, d = 10 genes). (a) Posterior Similarity Matrix (PSM), showing the estimated probability that pairs of cells belong to the same cluster. The block-diagonal structure indicates the presence of seven well-separated clusters (two large and five small equal clusters). (b) Heatmap comparing the inferred clustering from the same MCMC chain to the known cell type annotations from the scRNA-seq dataset.}
    \label{fig:matrixGENE}
\end{figure}

\subsection{Section \ref{sec:real-pam50}}

In Figure \ref{fig:PAMData} below we show the PSM and heat map obtained from our analysis of the PAM50 data. PAM50 is a 50-gene signature used to classify breast cancer into five intrinsic molecular subtypes: \emph{Luminal A}, \emph{Luminal B}, \emph{HER2-enriched}, \emph{Basal-like}, and \emph{Normal-like} \citep{perou2000molecular, sorlie2001gene, parker2009supervised}. The dataset has $n = 348$ samples and $d = 50$ genes. We model it with a mixture with $K^{\ast} = K=5$ and spherical covariance matrix.

\begin{figure}[!htb]
    \centering
    \begin{subfigure}[t]{\textwidth}
        \includegraphics[scale = 0.42]{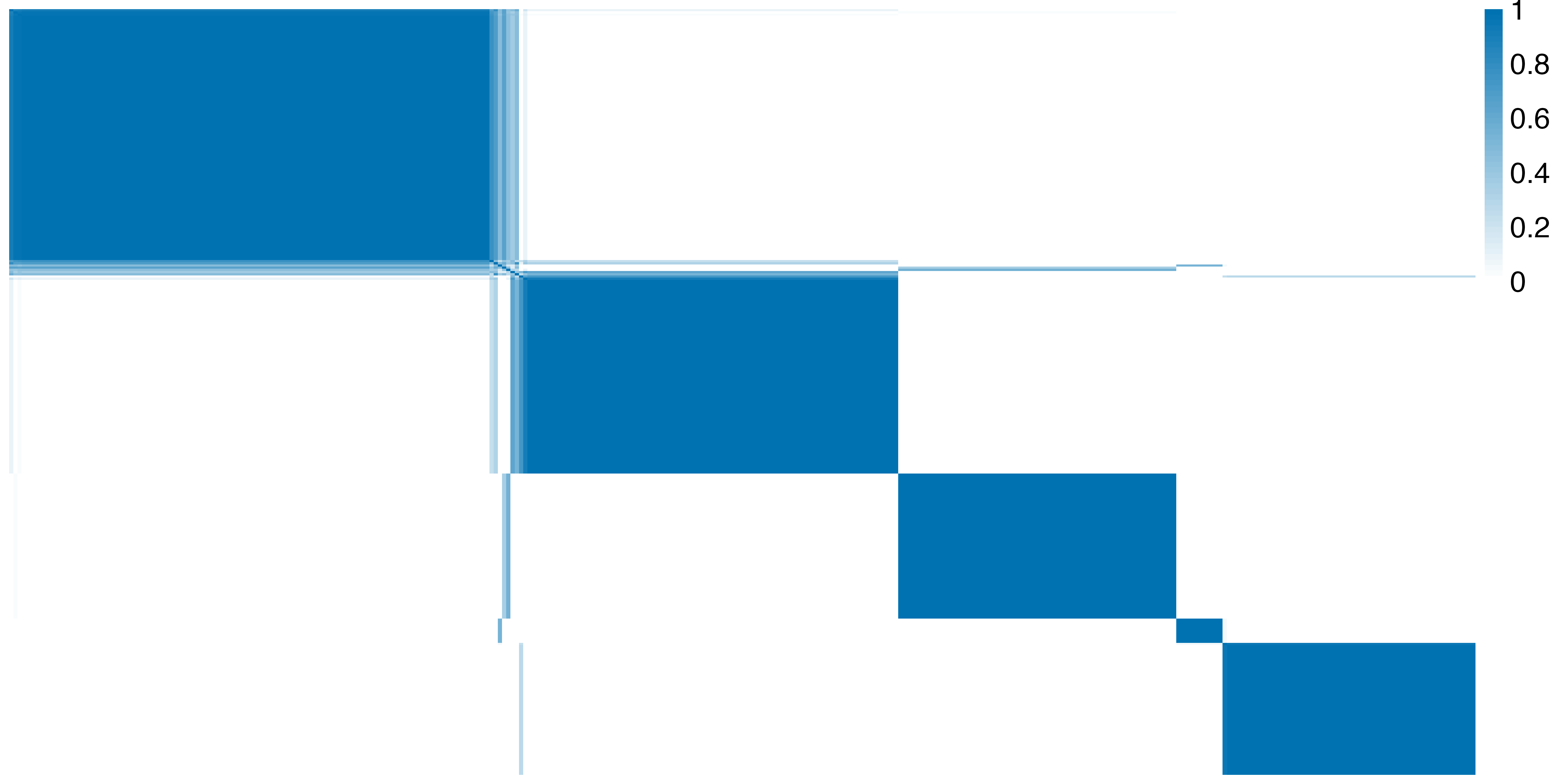} 
        \caption{Posterior similarity matrix.}
        \label{fig:psm_pam}
    \end{subfigure}
    
    \vspace{5em}
    
    \begin{subfigure}[t]{\textwidth}
        \includegraphics[scale = 0.48]{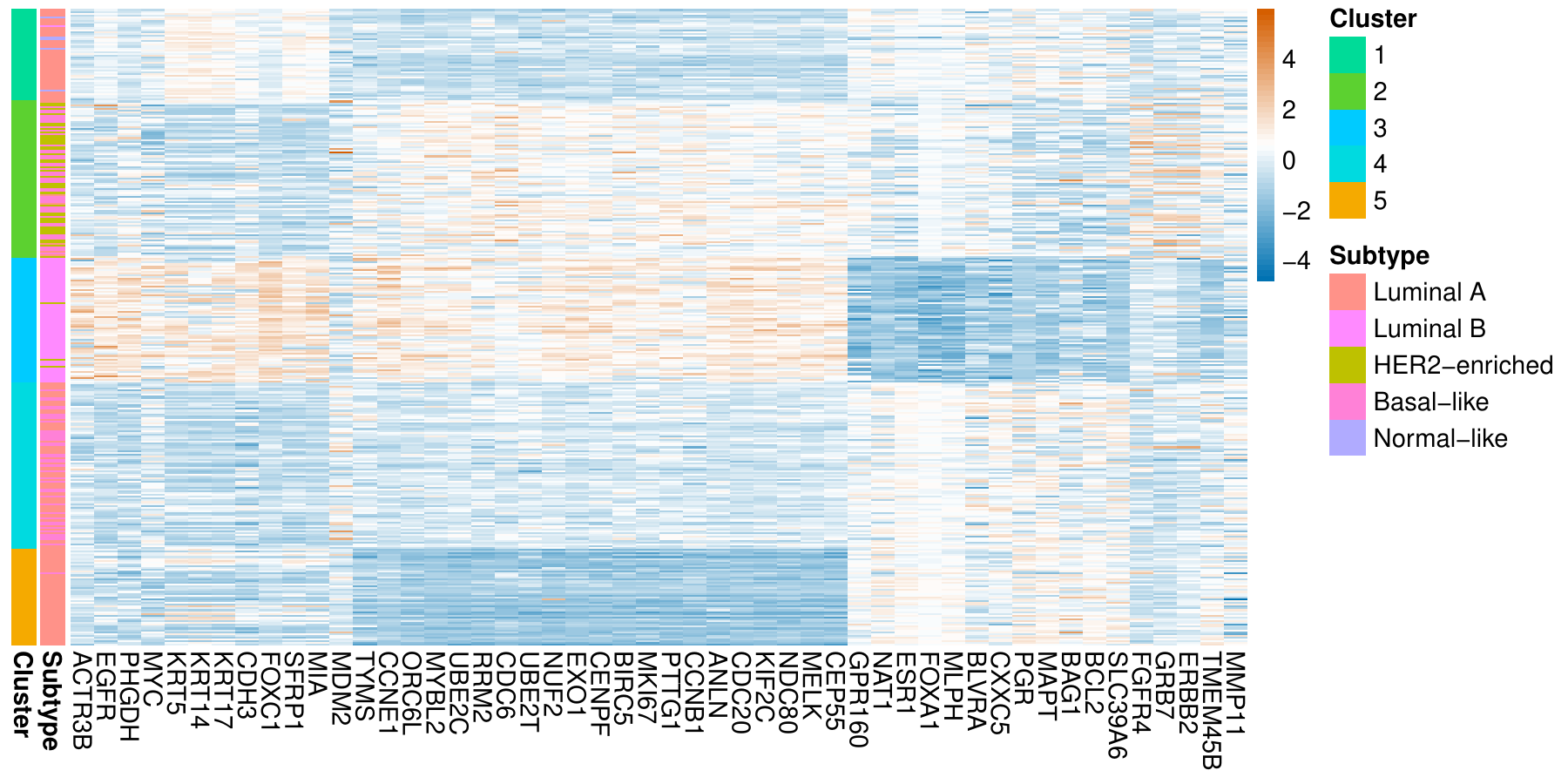} 
        \caption{Gene expression heatmap.}
        \label{fig:pheat_pam}
    \end{subfigure}

    \vspace{4em}
    
    \caption{PAM50 breast cancer gene expression clustering results (n = 348 tumor samples, d = 50 genes). (a) Posterior Similarity Matrix (PSM), representing the estimated probability that pairs of samples are assigned to the same cluster (four large and one small). The clear block-diagonal pattern suggests five well-separated clusters, corresponding to the intrinsic subtype structure. (b) Heatmap comparing the inferred clustering from the same MCMC chain to the known PAM50 molecular subtypes.}
    \label{fig:PAMData}
\end{figure}

\newpage

\section{Stochastic approximation proof}\label{app:SA}

We recast our adaptive scheme into the stochastic approximation framework of \citet{Andrieu2005}. In particular, we analyse the algorithm from iteration $s$ onwards, where $\lambda_t$ is no longer being adapted and has fixed value $\lambda_t \equiv 1$. Then our algorithm can be written as:
\begin{equation*}
        \boldsymbol{\alpha}_t = \boldsymbol{\alpha}_{t-1} + \rho_{t} H(\boldsymbol{\alpha}_{t-1}, (\boldsymbol{\Theta}_{t-1}, \boldsymbol{\pi}_{t-1}, \boldsymbol{Z}_{t-1})),
\end{equation*}
where we make the following definitions:
\begin{equation*}
    \begin{split}
        H(\boldsymbol{\alpha}, (\boldsymbol{\Theta, \pi, Z}))&:= \boldsymbol{D}(\mathbf{\Theta, \pi, Z}) - \boldsymbol{\alpha},\\
        h(\boldsymbol{\alpha})&:= \mathbb E_{\nu^{\ast}}[\boldsymbol{D}]-\boldsymbol{\alpha},\\
        \xi_t &:= \boldsymbol{D}(\boldsymbol{\Theta}_{t-1}, \boldsymbol{\pi}_{t-1}, \boldsymbol{Z}_{t-1}) -\mathbb E_{\nu^{\ast}}[\boldsymbol{D}],\quad \text{for all } t \in \mathbb N.
    \end{split}
\end{equation*}
In this setting, we observe that the function $h$ is in fact affine, and hence can be expressed as
\begin{equation}
        h(\balpha)=-\nabla J(\balpha),\quad J(\balpha)=-\frac{1}{2}\|\mathbb E_{\nu^{\ast}}[\boldsymbol D] - \boldsymbol\alpha\|^2.
        \label{eq:h_repn}
\end{equation}
Then looking at the assumption \cite[(A1)]{Andrieu2005}, we clearly see that the set of stationary points for $\balpha$ is the singleton
\begin{equation*}
    \mathcal L = \mathbb E_{\nu^{\ast}}[\boldsymbol{D}] \subset \mathbb R^n  .
\end{equation*}
And thus we expect this to be the almost-sure limit of $\balpha_t$.\\
To prove Theorem~\ref{thm:SA_conv} we will make use of \cite[Theorem~5.5]{Andrieu2005}. In our situation, it is considerably simpler than the full generality of \cite[Theorem~5.5]{Andrieu2005}, since $\balpha$ is constrained to a compact set; see Lemma~\ref{lemma:alpha_compact} below.
We need to check that the conditions (A1)-(A4) of \cite[Theorem~5.5]{Andrieu2005} are satisfied.

\begin{lemma}
    For our adaptive Gibbs chain, conditions (A1), (A2) and (A4) of \citet{Andrieu2005} hold.
\end{lemma}
\begin{proof}
    By the above discussion, \eqref{eq:h_repn}, and the comments of \citet{Andrieu2005} immediately following (A1), the condition (A1) holds. (A2) holds since for each $\balpha$ (which by construction cannot be 0 in any component), the invariant distribution of a given Gibbs sampler $P_{\balpha}$ for $(\boldsymbol{\Theta}, \boldsymbol\pi, \boldsymbol{Z})$ is always $\nu^{\ast}$. 
    Finally (A4) is automatically satisfied by our choice of $\rho_t$; see the discussion in \citet{Andrieu2005} immediately following (A4).
\end{proof}

\noindent In order to establish (A3), we first state a standard fact of finite Gaussian mixture models.
\begin{lemma}\label{lemma:cond_indep}
    For our finite mixture model, it holds that $\bm{\Theta}\mid\bm{Z},\bm{X}$ and $\bm{\pi}\mid\bm{Z},\bm{X}$ are conditionally independent.
\end{lemma}
\noindent This can be easily verified as the full conditional distribution $\bm{\Theta},\boldsymbol{\pi}\mid\bm{Z},\bm{X}$ can be factorised into the product of the two (conditional) marginals.
For our Gibbs sampler, we have the important following corollary.

\begin{lemma}\label{lemma:two_comp}
    For a fixed $\boldsymbol{\alpha}$, our Gibbs sampler $P_{\balpha}$ on $(\boldsymbol{Z},\boldsymbol{\Theta}, \boldsymbol{\pi})$ can be interpreted as a \textit{two-component} Gibbs sampler.
\end{lemma}
\begin{proof}
    Consider one full update, as in steps 4-9 of Algorithm~\ref{alg:ARSG4}. Since we have conditional independence of $\bm{\Theta}\mid\bm{Z},\bm{X}$ and $\bm{\pi}\mid\bm{Z},\bm{X}$ from Lemma~\ref{lemma:cond_indep}, we can interpret steps~8-9 of Algorithm~\ref{alg:ARSG4} as jointly sampling $\boldsymbol{\pi}, \boldsymbol{\Theta}$ from the joint distribution $\boldsymbol{\pi,\Theta}\mid\boldsymbol{Z},\boldsymbol{X}$.
\end{proof}

\begin{lemma}
    Condition (DRI) of \citet{Andrieu2005} holds for our adaptive Gibbs chain. In turn, this implies that (A3) holds.
\end{lemma}
\begin{proof}
    Proposition~6.1 of \citet{Andrieu2005} shows that (DRI) implies (A3), so we focus on (DRI). Since $\balpha$ is confined to a compact state space (since $D$ is bounded), and in our Gaussian mixture case we have the Feller property (see, e.g. \citet{Hobert1998}), it is sufficient to establish drift and minorisation for a given $P_{\balpha}$.\\
    We establish this using the fact that the chain $P_{\balpha}$ alternately samples from two fixed conditional kernels by Lemma~\ref{lemma:two_comp}. This property implies that the convergence properties of the overall chain are the same as those of the marginal chain. See for instance \citet{diebolt1994estimation}, where this is referred to \textit{duality} and \citet[Section~3.3]{Roberts2001} where it is referred to as \textit{de-initialising}.\\
    In particular, since the $\boldsymbol{Z}$ variables are on a compact state space, their corresponding marginal chain is uniformly ergodic. Hence the overall joint chain is uniformly ergodic, and a drift condition must hold, e.g. by \citet[Theorem~15.0.1]{Meyn1993a}.
\end{proof}

\noindent Taken together, we have established the prerequisites of \cite[Theorem~5.5]{Andrieu2005}, from which we conclude that the desired convergence holds.

\begin{lemma}
    The vectors $\bm{\alpha}$ take values in a fixed compact subset of $\mathbb R^d$.
    \label{lemma:alpha_compact}
\end{lemma}
\begin{proof} We aim to prove that the sequence of vectors \(\bm{\alpha}_t \in \mathbb{R}^n\), defined recursively, lies in a compact subset of \(\mathbb{R}^n\): we show that the sequence \(\bm{\alpha}_t\) lives in a bounded set. Let \(\bm{\alpha}_t \in \mathbb{R}^n\) be defined recursively as:
\begin{equation}
    \bm{\alpha}_t = f(t) \bm{\alpha}_{t-1} + g(t) \bm{D}_t^{\lambda_t},
    \label{eq:recurrence}
\end{equation}
where \(\bm{D}_t^{\lambda_t} \in \mathbb{R}^n\) is a vector that depends on a probability allocation \(\bm{p}_t \in [0,1]^n\) and the decay parameter \(\lambda_t \in [1, \Lambda] \subset \mathbb{R}_+\), with \(\Lambda > 1\) a fixed user-defined upper bound. In particular, each component of \(\bm{D}_t^{\lambda_t}\) is defined as $D_{i,t}^{\lambda_t} = \exp\left(-\lambda_t p_{i,t}\right)$, which implies that \(\bm{D}_t^{\lambda_t} \in [\exp(-\Lambda), 1]^n\). Therefore, for every \(t\), the vector \(\bm{D}_t^{\lambda_t}\) lies in a compact subset of \(\mathbb{R}^n\). \\
The functions \(f(t)\) and \(g(t)\) are defined according to two distinct regimes:
\begin{equation*}
f(t) =
\begin{cases}
\frac{1}{2} \left(1 + \tanh\left(\frac{t - s}{a}\right)\right), & t \leq s \\
\frac{t - s + 1}{t - s + 2}, & t > s
\end{cases}, \qquad
g(t) =
\begin{cases}
\frac{1}{2} \left(1 - \tanh\left(\frac{t - s}{a}\right)\right), & t \leq s \\
\frac{1}{t - s + 2}, & t > s
\end{cases}
\end{equation*}

\noindent where \(s \in \mathbb{N}\) and $a \in \mathbb{R}_{+}$. The functions \(f(t)\) and \(g(t)\) are scalar weights that depend on \(t\) and satisfy the condition \(f(t) + g(t) = 1\) for all \(t\). As a result, each update \(\bm{\alpha}_t\) in Equation~\eqref{eq:recurrence} is a convex combination of \(\bm{\alpha}_{t-1}\) and \(\bm{D}_t^{\lambda_t}\). By unrolling the recursion, we obtain for any \(t \geq 1\):
\[
\bm{\alpha}_t = \omega_{t,0} \bm{\alpha}_0 + \sum_{k=1}^t \omega_{t,k} \bm{D}_k^{\lambda_k},
\]
where the weights \(\omega_{t,k}\) are defined recursively as:
\begin{equation*}
    \omega_{t,k} = g(k) \prod_{j=k+1}^{t} f(j), \quad \text{for } 1 \leq k \leq t \quad \text{and } \, 
    \omega_{t,0} = \prod_{j=1}^{t} f(j).
\end{equation*}
Each \(\omega_{t,k} \geq 0\), and the sum satisfies:
\[
\sum_{k=0}^t \omega_{t,k} = 1,
\]
so the vector \(\bm{\alpha}_t\) is a convex combination of \(\bm{\alpha}_0\) and the vectors \(\bm{D}_k^{\lambda_k}\), for \(k = 1, \ldots, t\). Since $\bm{\alpha}_t \in \text{conv} \left( \{ \bm{\alpha}_0, \bm{D}_1^{\lambda_1}, \ldots, \bm{D}_t^{\lambda_t} \} \right)$, and each \(\bm{D}_k^{\lambda_k} \in [\exp(-\Lambda), 1]^n\), which is a compact set in \(\mathbb{R}^n\), we conclude that \(\bm{\alpha}_t\) lies in the convex hull of a finite number of points from a compact set---hence, \(\bm{\alpha}_t\) itself lies in a compact set.
This holds both in the regime \(t \leq s\), where weights are smoothly varying, and in the regime \(t > s\), where weights follow a polynomial decay.
\end{proof}


Figure \ref{fig:resultConv} provides empirical evidence for this convergence, showing the spatial pattern of convergence differences diminishing over time for the Miller-Harrison dataset.

\begin{figure}[!ht]
    \centering
    \includegraphics[scale = 0.82]{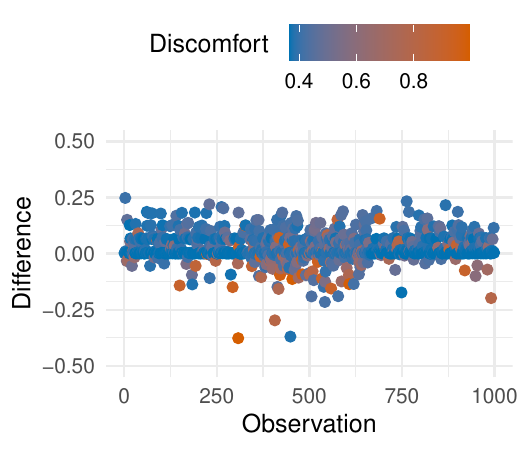}
    \includegraphics[scale = 0.82]{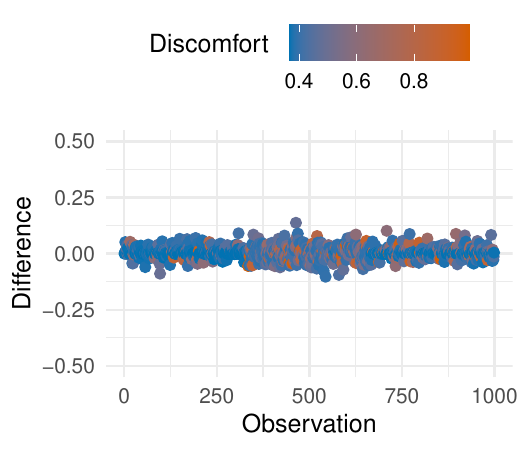}
    \caption{Miller-Harrison dataset with $n = 1000$ and $d = 2$: Empirical convergence of the selection probabilities. Left panel shows the elementwise differences between $\boldsymbol{\alpha}_t$ and $\mathbb E_{\nu^{\ast}}[\boldsymbol{D}^1]$ at iteration 1000; right panel shows the same at iteration 10000. The convergence is evident as the differences approach zero.}
    \label{fig:resultConv}
\end{figure}

\end{document}